\documentclass[preprint,12pt,authoryear]{elsarticle}

\usepackage{amssymb}
\usepackage{amsthm}

\usepackage{onimage}
\usepackage{amsmath}
\usepackage{tabstackengine}
\usepackage{mathtools,tikz}
\usepackage{graphicx}
\usepackage{epstopdf, epsfig}
\usepackage{scalerel,stackengine}
\usepackage{subfig}
\usepackage{float}
\usepackage{bm}
\usepackage{tikz}
\usetikzlibrary{shapes,arrows,calc,positioning}
\usepackage{algorithm,algpseudocode}
\newtheorem{proposition}{Proposition}

\tikzset{
        block_small/.style = {draw, rectangle,
            minimum height=1cm,
            minimum width=1cm},
        block/.style = {draw, rectangle,
            minimum height=1cm,
            minimum width=2cm},
        input/.style = {coordinate,node distance=1cm},
        output/.style = {coordinate,node distance=4cm},
        arrow/.style={draw, -latex,node distance=2cm},
        pinstyle/.style = {pin edge={latex-, black,node distance=2cm}},
        sum/.style = {draw, circle, node distance=1cm},
}

\stackMath
\newcommand\reallywidehat[1]{%
\savestack{\tmpbox}{\stretchto{%
  \scaleto{%
    \scalerel*[\widthof{\ensuremath{#1}}]{\kern-.6pt\bigwedge\kern-.6pt}%
    {\rule[-\textheight/2]{1ex}{\textheight}}
  }{\textheight}%
}{0.5ex}}%
\stackon[1pt]{#1}{\tmpbox}%
}

\algtext*{EndFor}
\algtext*{EndIf}

\newcommand{\mat}[1]{{\bm{#1}}}
\newcommand{\mA}{\mat{A}}
\newcommand{\mB}{\mat{B}}
\newcommand{\mC}{\mat{C}}
\newcommand{\mF}{\mat{F}}
\newcommand{\mG}{\mat{G}}
\newcommand{\mH}{\mat{H}}

\newcommand{\mI}{\mat{I}}
\newcommand{\mK}{\mat{K}}
\newcommand{\mL}{\mat{L}}

\newcommand{\mP}{\mat{P}}
\newcommand{\mQ}{\mat{Q}}
\newcommand{\mR}{\mat{R}}
\newcommand{\mU}{\mat{U}}
\newcommand{\mV}{\mat{V}}
\newcommand{\mW}{\mat{W}}
\newcommand{\mT}{\mat{T}}

\newcommand{\mX}{\mat{X}}
\newcommand{\mY}{\mat{Y}}
\newcommand{\mJ}{\mat{J}}
\newcommand{\mZ}{\mat{Z}}
\newcommand{\mSigma}{\mat{\pmb{\Sigma}}}
\newcommand{\mPhi}{\mat{\pmb{\Phi}}}
\newcommand{\mPsi}{\mat{\pmb{\Psi}}}

\newcommand{\vect}[1]{{\bm{#1}}}

\newcommand{\vn}{\vect{n}}
\newcommand{\vf}{\vect{f}}

\newcommand{\vq}{\vect{q}}

\newcommand{\vu}{\vect{u}}

\newcommand{\vw}{\vect{w}}
\newcommand{\vx}{\vect{x}}
\newcommand{\vy}{\vect{y}}
\newcommand{\vz}{\vect{z}}
\newcommand{\vX}{\vect{X}}

\newcommand{\vxi}{\vect{\xi}}

\newcommand{\ddt}{\frac{\mathrm{d}}{\mathrm{d}t}}

\newcommand\Algphase[1]{%
\vspace*{-.5\baselineskip}\Statex\hspace*{\dimexpr-\algorithmicindent-2pt\relax}\rule{\textwidth}{0.4pt}%
\Statex\hspace*{-\algorithmicindent}\textbf{#1}%
\vspace*{-.5\baselineskip}\Statex\hspace*{\dimexpr-\algorithmicindent-2pt\relax}\rule{\textwidth}{0.4pt}%
}

\stackMath
\setstackgap{L}{15pt}
\setstacktabbedgap{2pt}
\def\lrgap{\kern6pt}
\def\xbracketVectorstack#1{\left[\lrgap\Vectorstack{#1}\lrgap\right]}
\def\xbracketMatrixstack#1{\left[\lrgap\tabbedCenterstack{#1}\lrgap\right]}

\begin{document}

\begin{frontmatter}



\title{Continuous-time balanced truncation for time-periodic fluid flows using frequential Gramians}


\author{Alberto Padovan\corref{cor1}}
\ead{apadovan@princeton.edu}
\author{Clarence W. Rowley}
\ead{cwrowley@princeton.edu}
\cortext[cor1]{Corresponding author}
\affiliation{organization={Mechanical and Aerospace Engineering Department, Princeton University},
addressline={Olden St.},
postcode={08544},
city={Princeton},
state={NJ},
country={USA}}

\begin{abstract}
Reduced-order models for flows that exhibit time-periodic behavior (e.g., flows in turbomachinery and wake flows) are critical for several tasks, including active control and optimization.
One well-known procedure to obtain the desired reduced-order model in the proximity of a periodic solution of the governing equations is continuous-time balanced truncation.
Within this framework, the periodic reachability and observability Gramians are usually estimated numerically via quadrature using the forward and adjoint post-transient response to impulses.
However, this procedure can be computationally expensive, especially in the presence of slowly-decaying transients.
Moreover, it can only be performed if the periodic orbit is stable in the sense of Floquet.
In order to address these issues, we use the frequency-domain representation of the Gramians, which we henceforth refer to as \emph{frequential Gramians}.
First, these frequential Gramians are well-defined for both stable and unstable dynamics.
In particular, we show that when the underlying system is unstable, these Gramians satisfy a pair of allied differential Lyapunov equations.
Second, they can be estimated numerically by solving algebraic systems of equations that lend themselves to heavy computational parallelism and that deliver the desired post-transient response without having to follow physical transients.
We demonstrate the method on a periodically-forced axisymmetric jet at Reynolds numbers $Re=1250$ and $Re = 1500$.
At the lower Reynolds number, the flow strongly amplifies subharmonic perturbations and exhibits vortex pairing about a Floquet-stable $T$-periodic solution.
At the higher Reynolds number, the underlying $T$-periodic orbit is unstable and the flow naturally settles onto a $2T$-periodic limit cycle characterized by pairing vortices.
At both Reynolds numbers, we compute a reduced-order model and we use it to design a feedback controller and a state estimator capable of suppressing vortex pairing.
\end{abstract}


\end{frontmatter}



\section{Introduction}
\label{sec:intro}

Physical processes are often governed by partial differential equations, which, upon spatial discretization, lead to high-dimensional systems of ordinary differential equations.
Although recent advances in computational resources have allowed us to simulate these systems quite efficiently, tasks such as controller design and optimization can seldom be performed in the original high-dimensional space.
It therefore becomes necessary to develop low-order models that capture the salient features of the underlying dynamics.
In this paper, we focus on systems that exhibit time-periodic behavior, and we seek a low-order representation of the dynamics in the proximity of a time-periodic solution of the governing equations.

While there are many existing methods for model reduction of both linear and nonlinear systems, here we provide an overview of the linear techniques based on ``balancing.''
In the simplest of cases, i.e., the balancing of a linear time-invariant system,
one seeks a reduced-order model by first identifying a change of coordinates
where two matrices known as reachability and observability Gramians are equal and diagonal.
This method, initially introduced by \cite{Moore1981principal}, has become increasingly popular because of well-known a-priori error bounds (see, e.g., \cite{dullerud}) and because of its relatively low computational cost.
More recently, the balanced proper orthogonal decomposition (BPOD) framework introduced by \cite{Rowley2005} led to a further reduction of the computational cost in systems with a large number of outputs, and the method has since become a benchmark for model reduction of linear systems as well as of nonlinear systems that evolve near a steady state.
Although balanced truncation and BPOD were originally conceived for stable systems, they have also been applied to unstable systems upon slight modifications.
For instance, \cite{ahuja_rowley_2010} proposed splitting the stable and
unstable eigenspaces, and balancing the stable dynamics, while treating the unstable eigenspace exactly.
Alternatively, \cite{dergham} obtained balanced reduced-order models of an open cavity flow with an underlying unstable steady state by leveraging the frequency-domain representation of the Gramians, which is well-defined for both stable and unstable systems (see, e.g., \cite{godunov}).
Finally, \cite{flinois2015} showed that the original algorithm developed for stable systems could be used to balanced unstable systems without the a-priori splitting of the stable and unstable eigenspaces.

Balanced truncation has also been used for discrete-time periodic systems (see, for instance, \cite{longhi1999}, \cite{varga2000} and \cite{farhood}).
More recently, a procedure similar to BPOD was developed by \cite{Ma} by lifting the discrete-time periodic system into a higher-dimensional linear time-invariant system.
That formulation was then applied in \cite{MaThesis} for controller design to stabilize an unstable periodic orbit in the wake of a flat plate.
The balancing of linear time-varying systems in their continuous-time formulation is discussed in \cite{sandberg2004}, and error bounds for monotonically-balanced and non-monotonically-balanced systems are presented.
Continuous-time balancing was also performed in \cite{lang2016}, where the
authors presented an implicit time integration method to solve the differential Lyapunov equations that govern the dynamics of the (time-varying) reachability and observability Gramians.
Despite the fact that balancing for periodic systems is well-understood, to the best of the our knowledge it is rarely used in practice in very high-dimensional systems such as two-dimensional or three-dimensional fluid flows.
In fact, the (discrete-time) application in \cite{MaThesis} is the only one we are aware of.
This is most likely due to the fact that computing the Gramians for a time-varying system can be expensive.
In particular, this requires computing the post-transient response to forward and adjoint impulses, and the computational cost can grow significantly if the underlying dynamics exhibit slowly-decaying transients.
Furthermore, this procedure to estimate the Gramians can only be performed on systems that are stable in the sense of Floquet, unless the stable and unstable Floquet eigenspaces are treated separately (as in \cite{Ma}) at additional computational cost.
Here, we propose to address these problems using the frequency-domain representation of the Gramians.
We henceforth refer to these Gramians as \emph{frequential Gramians}.

As in time-invariant systems, if the underlying system is stable, the frequential Gramians agree with the time-domain representation of the Gramians.
Unlike their time-domain counterparts, the frequential Gramians are also well-defined if the dynamics are unstable, and we show that they satisfy a pair of allied differential Lyapunov equations.
Consequently, we can use these Gramians to obtain balanced low-order models of unstable systems, and this is particularly important if we wish to design reduced-order stabilizing controllers.
We shall also see that while the frequential Gramians for time-invariant systems are defined in terms of the resolvent operator associated with the underlying system, the frequential Gramians for time-periodic systems are defined in terms of the harmonic resolvent operator \citep{padovan2020,padovan2022}.

From a computational standpoint, the use of frequential Gramians can lead to computational savings.
Specifically, estimating the Gramians no longer requires performing impulse responses in the time domain, but it simply amounts to solving algebraic systems of equations that lend themselves to heavy computational parallelism and that deliver the desired post-transient solution without having to follow the physical transients.
As discussed in section \ref{sec:algorithm_for_gramians}, additional savings can be obtained by leveraging some of the symmetries of the harmonic resolvent operator, which, as previously mentioned, is used to define the frequential Gramians.
More thorough computational considerations are presented in section \ref{subsec:computational_considerations}.


We use this framework to compute reduced-order models for a periodically-forced axisymmetric incompressible jet at Reynolds numbers $Re = 1250$ and $Re = 1500$.
At $Re = 1250$, the flow exhibits a Floquet-stable periodic orbit of period $T$, characterized by an unpaired vortex street.
However, as discussed in \cite{ardali_jet19} and \cite{padovan2022}, this configuration is extremely sensitive to period-doubling perturbations, so that any small-amplitude perturbation will cause neighboring vortices to merge and pair.
Here, we compute a reduced-order model of the dynamics in the proximity of the
periodic orbit, and we then design a feedback controller and an observer that
successfully suppress vortex pairing in the presence of disturbances.
At $Re = 1500$ the underlying $T$-periodic orbit is linearly unstable, so any small perturbation will grow and eventually settle onto a $2T$-periodic limit cycle characterized by pairing vortices.
As before, we compute a reduced-order model and we design a controller and an observer to restabilize the $T$-periodic orbit and suppress vortex pairing.

Although, to the best of our knowledge, this is the first time that frequential Gramians are used in the balancing of time-periodic systems, it is important to mention related work that leverages the frequency-domain representation of time-periodic systems.
For instance, \cite{MRJ2008} proposed solving a sequence of simplified Sylvester and Lyapunov equations to approximate the $\mathcal{H}_2$ norm of linear periodically time-varying system, where the periodic component is small.
Their approach was then implemented on a pressure-driven channel subject to streamwise oscillations of the bottom wall.
The same approach was used in \cite{moarref2010} and \cite{moarref2012} to design controllers to suppress the onset of turbulence in a channel, and to achieve turbulent drag reduction in a channel, respectively.
Similar tools were also used in \cite{ran2021} to design spanwise-periodic riblets with the objective of reducing drag in a turbulent channel.
For a thorough overview of frequency-domain methods for the analysis of fluids flows we refer to the review paper by \cite{jovanovic2021} and references therein.

\section{Frequential reachability and observability Gramians}
\label{sec:balancing_stable}

In continuous-time balanced truncation for time-periodic systems, one seeks a
time-periodic reduced-order model by first identifying a change of coordinates
that simultaneously diagonalizes the time-periodic reachability and
observability Gramians, which will defined below.
While the balancing procedure will be discussed in detail in section \ref{sec:bt}, for now, it suffices to say that computing the Gramians (or their factorization) is the most computationally expensive step.
In this section we therefore focus on the computation of the Gramians and we show that, similarly to the linear time-invariant case, these can be defined in the frequency domain.
We also show that the frequency-domain representation of the Gramians is well
defined even when the underlying dynamics are unstable; more specifically, the frequential Gramians satisfy a pair of differential Lyapunov equations.


\subsection{Preliminaries}
We begin by considering a linear time-periodic system with state $\vx(t)\in\mathbb{R}^N$, control input $\vu(t)\in\mathbb{R}^M$, and output $\vy(t) \in \mathbb{R}^Q$
\begin{equation}
    \label{eq:sys_ltp}
    \begin{aligned}
        \ddt \vx(t) &= \mA(t) \vx(t) + \mB(t)\vu(t) \\
        \vy(t) &= \mC(t) \vx(t),
    \end{aligned}
\end{equation}
where the linear operators $\mA(t)$, $\mB(t)$ and $\mC(t)$ are all periodic with period~$T$ (i.e., $\mA(t) = \mA(t+T)$).
The system \eqref{eq:sys_ltp} arises in fluid mechanics when the Navier-Stokes equations are linearized about a $T$-periodic solution.
In incompressible flow, the state $\vx(t)$ may be taken as the divergence-free velocity field at the cell faces (or cell centers) of a computational grid, the forcing term $\mB(t)\vu(t)$ may be understood as a volumetric or boundary input and the output $\vy(t)$ could be some measured output (e.g., the velocity at some desired physical location in the flow).
Usually, the operators $\mB(t)$ and $\mC(t)$ are time-invariant (e.g., if the control input and the measured output are located at some fixed physical coordinate), but here we include time dependence for the sake of generality.

In order to derive the frequency-domain representation of the reachability and observability Gramians associated with \eqref{eq:sys_ltp}, it is notationally convenient to first diagonalize \eqref{eq:sys_ltp} via a Floquet change of coordinates \citep{floquet}.
In particular, by Floquet's theorem, there exists a possibly complex $T$-periodic change of coordinates $\vx(t) = \mV(t)\vz(t)$ such that
\begin{equation}
\label{eq:sys_ltp_lti}
    \begin{aligned}
    \ddt \vz(t) &= \mJ \vz(t) + \underbrace{\mW(t)^*\mB(t)}_{\widetilde{\mB}(t)}\vu(t) \\
    \vy(t) &= \underbrace{\mC(t)\mV(t)}_{\widetilde\mC(t)}\vz(t),
    \end{aligned}
\end{equation}
where $\mW(t)^*\mV(t) = \mI$ for all $t$ and $\mJ$ is a diagonal time-invariant matrix containing the Floquet exponents associated with \eqref{eq:sys_ltp}.
Here, $\mW^*$ denotes the Hermitian transpose of $\mW$.
It is easy to verify that the periodic orbit is stable if and only if all the Floquet exponents lie in the left-half plane, and it is unstable otherwise.


\subsection{Stable dynamics}
\label{subsec:stable_dynamics}
In this subsection we assume that \eqref{eq:sys_ltp} is stable in the sense just
described.
The reachability and observability Gramians associated with the diagonalized dynamics \eqref{eq:sys_ltp_lti} may be defined as
\begin{align}
    \mG_{R}(t_0,t) &= \int_{t_0}^t e^{\mJ(t-\tau)}\widetilde{\mB}(\tau)\widetilde{\mB}(\tau)^*e^{\mJ^*(t-\tau)}\mathrm{d}\tau \label{eq:Gr} \\
    \mG_{O}(t,t_f) &= \int_{t}^{t_f}e^{\mJ^*(\tau-t)}\widetilde{\mC}(\tau)^*\widetilde{\mC}(\tau)e^{\mJ(\tau-t)}\mathrm{d}\tau. \label{eq:Go}
\end{align}
It can be shown that as $t_0\to-\infty$ and $t_f\to+\infty$,
the Gramians $\mG_R(t_0,t)$ and $\mG_O(t,t_f)$ are periodic functions of~$t$ with period~$T$.
This is a well-known result that can be illustrated by considering the forward and adjoint differential Lyapunov equations below
\begin{align}
    \ddt \mP(t) &= \mJ\mP(t) + \mP(t)\mJ^* + \widetilde{\mB}(t)\widetilde{\mB}(t)^* \label{eq:fwd_lyap} \\
    -\ddt \mQ(t) &= \mJ^*\mQ(t) + \mQ(t)\mJ + \widetilde{\mC}(t)^*\widetilde{\mC}(t). \label{eq:adj_lyap}
\end{align}
It can be readily checked that the solution $\mP(t)$ of \eqref{eq:fwd_lyap} may be written as
\begin{equation}
    \label{eq:sol_fwd_lyap}
    \mP(t) = \underbrace{e^{\mJ (t-t_0)}\mP(t_0)e^{\mJ^* (t-t_0)}}_{\text{initial cond. response}} + \underbrace{\mG_{R}(t_0,t)}_{\text{forced resp.}},
\end{equation}
where $\mP(t_0)$ is the initial condition at time $t = t_0$.
By Theorem 20 in \cite{bolzern88}, equation \eqref{eq:fwd_lyap} admits a unique positive-definite (for all times) $T$-periodic solution if \eqref{eq:sys_ltp_lti} is stable and controllable.
This solution may be understood as the long-time response of \eqref{eq:fwd_lyap} to the external forcing $\widetilde{\mB}(t)\widetilde{\mB}(t)^*$.
In particular, if \eqref{eq:sys_ltp} is stable, it is clear that for $t\gg t_0$ the initial condition response in \eqref{eq:sol_fwd_lyap} will go to zero and we will be left with the forced response $\mP(t)=\mG_{R}(t_0,t)$.
A similar argument holds for the adjoint differential Lyapunov equation
\eqref{eq:adj_lyap} for $t\ll t_f$.

Now that we have established the relationship between the Gramians and the periodic solution of the corresponding differential Lyapunov equations, we can seek this solution in the frequency domain.
In particular, since the solution is periodic with period $T$, we can write
\begin{equation}
\label{eq:fourier_mP_mQ}
    \mP(t) = \sum_{k\in\mathbb{Z}} \mP_k e ^{i k \omega t},\quad \mQ(t) = \sum_{k\in\mathbb{Z}} \mQ_k e^{i k \omega t},\quad \omega = \frac{2\pi}{T}.
\end{equation}
Moreover, since $\widetilde\mB(t)$ and $\widetilde\mC(t)$ are also $T$-periodic, they can be written in a Fourier series analogous to the ones in \eqref{eq:fourier_mP_mQ}.
Substitution into the corresponding Lyapunov equations leads, for a fixed integer $k$, to
\begin{align}
\left(-i k \omega\mI + \mJ\right) \mP_{k} +  \mP_{k} \mJ^* + \sum_{l\in\mathbb{Z}} {\widetilde\mB}_{k+l} {\widetilde\mB}_l^* &= 0 \label{eq:sylv_fwd}\\
\left(i k \omega\mI + \mJ^*\right) \mQ_{k} +  \mQ_{k} \mJ + \sum_{l\in\mathbb{Z}} {\widetilde\mC}^*_{-l-k} {\widetilde\mC}_{-l} &= 0. \label{eq:sylv_adj}
\end{align}
It is easy to verify (see, e.g., \cite{godunov}) that the solution of the two algebraic Sylvester equations \eqref{eq:sylv_fwd} and \eqref{eq:sylv_adj} is given by
\begin{align}
\mP_{k} &= \frac{1}{2\pi}\int_{-\infty}^\infty
\left(i\gamma\mI - (-i k\omega\mI + \mJ)\right)^{-1} \sum_{l\in\mathbb{Z}} \widetilde\mB_{k+l} \widetilde\mB^*_l\left(-i\gamma\mI - \mJ^*\right)^{-1}\,\mathrm{d}\gamma \label{eq:Pkmj}\\
 \mQ_{k} &= \frac{1}{2\pi}\int_{-\infty}^\infty
\left(-i\gamma\mI - (i k\omega\mI + \mJ^*)\right)^{-1} \sum_{l\in\mathbb{Z}} \widetilde\mC^*_{-l-k} \widetilde\mC_{-l}\left(i\gamma\mI - \mJ\right)^{-1}\,\mathrm{d}\gamma. \label{eq:Qkmj}
\end{align}
It can be shown that the above integrals converge for all $k$ as long as no Floquet
exponent (i.e., eigenvalue of $\mJ$) lies on the imaginary axis.
This implies that $\mP_k$ and $\mQ_k$ are well-defined also when one or more of the Floquet exponents lie in the open right-half plane (i.e., when the periodic dynamics are unstable).
In the next subsection, we show that when the dynamics are unstable, the Gramians in \eqref{eq:Pkmj} and \eqref{eq:Qkmj} satisfy corresponding differential Lyapunov equations.

\subsection{Unstable dynamics}

In this subsection we assume that one or more of the eigenvalues of~$\mJ$ lie in the open right-half plane, so that the periodic dynamics are unstable.
We henceforth let $\mathcal{P}_s$ and $\mathcal{P}_u$ denote the projections onto the stable and unstable eigenspaces of $\mJ$.
Notice that since $\mJ$ is a diagonal matrix, $\mathcal{P}_s$ and $\mathcal{P}_u$ are orthogonal projections (and, in fact, diagonal themselves).
We can then state the following result.

\begin{proposition}
\label{prop:unstable_lyap_eqtn}
Suppose that no eigenvalue of~$\mJ$ lies on the imaginary axis.
Then the Fourier coefficients $ \mP_{k}$ and $ \mQ_{k}$ defined in \eqref{eq:Pkmj} and \eqref{eq:Qkmj} satisfy
the following Sylvester equations,
\begin{align}
\left(-i k \omega\mI + \mJ\right) \mP_{k} +  \mP_{k} \mJ^* + \mathcal{P}_s\sum_{l\in\mathbb{Z}} \widetilde\mB_{k+l} \widetilde\mB^*_l\mathcal{P}_s - \mathcal{P}_u\sum_{l\in\mathbb{Z}}\widetilde\mB_{k+l} \widetilde\mB^*_l\mathcal{P}_u &= 0 \label{eq:sylv_fwd_unstable}\\
\left(i k \omega\mI + \mJ^*\right) \mQ_{k} +  \mQ_{k} \mJ +\mathcal{P}_s\sum_{l\in\mathbb{Z}} \widetilde\mC^*_{-l-k} \widetilde\mC_{-l}\mathcal{P}_s - \mathcal{P}_u\sum_{l\in\mathbb{Z}}\widetilde\mC^*_{-l-k} \widetilde\mC_{-l}\mathcal{P}_u &= 0.  \label{eq:sylv_adj_unstable}
\end{align}
Via inverse Fourier transform, it follows that the $T$-periodic Gramians $\mP(t)$ and $\mQ(t)$
satisfy the differential Lyapunov equations below
\begin{align}
    \ddt \mP(t) &= \mJ\mP(t) + \mP(t)\mJ^* + \mathcal{P}_s\widetilde{\mB}(t)\widetilde{\mB}(t)^*\mathcal{P}_s - \mathcal{P}_u\widetilde{\mB}(t)\widetilde{\mB}(t)^*\mathcal{P}_u \label{eq:fwd_lyap_unstable} \\
    -\ddt \mQ(t) &= \mJ^*\mQ(t) + \mQ(t)\mJ + \mathcal{P}_s\widetilde{\mC}(t)^*\widetilde{\mC}(t)\mathcal{P}_s - \mathcal{P}_u\widetilde{\mC}(t)^*\widetilde{\mC}(t)\mathcal{P}_u. \label{eq:adj_lyap_unstable}
\end{align}
\end{proposition}

\begin{proof}
If all the Floquet exponents lie in the left-half plane, then $\mathcal{P}_u = \mathbf{0}$ and $\mathcal{P}_s = \mI$, and equations \eqref{eq:sylv_fwd_unstable} and \eqref{eq:sylv_adj_unstable} agree with \eqref{eq:sylv_fwd} and \eqref{eq:sylv_adj}, respectively.
In the general case when $\mathcal{P}_u \neq \mathbf{0}$, the proof is analogous to the linear time-invariant case in Section 10.2 of \cite{godunov}.
The second part of the proposition follows immediately via inverse Fourier transform.
\end{proof}

We have therefore established that the frequency-domain representation of the Gramians is well-defined for both stable and unstable dynamics.
By contrast, the time-domain integral representation of the Gramians is
well-defined only if the dynamics are stable; otherwise, the initial condition
response in~\eqref{eq:sol_fwd_lyap} would blow up for $t\gg t_0$.
In the next section we will take steps to make practical use of the results discussed in this section.


\section{Towards an efficient algorithm to compute the Gramians}
\label{sec:towards_efficient_algorithm}

In this section we address two main issues.
First and foremost, in balanced truncation we are not interested in explicitly computing the Fourier coefficients of the Gramians.
Instead, we would like to evaluate the Gramians at desired time $t\in[0,T)$.
Conveniently, we show in the upcoming subsection that the Gramians at any time $t$ can be computed as an outer product of frequency-domain quantities.
In the subsequent subsections, we address the second issue.
Namely, we would like to compute the Gramians in the physical coordinates $\vx(t)$.
While the diagonalized Floquet coordinates $\vz(t)$ have proven useful to understand the structure of the Gramians, they are not well-suited for computation.
Specifically, for high-dimensional systems, it is generally infeasible to compute the diagonalizing Floquet transformation that maps $\vx(t)$ coordinates to $\vz(t)$ coordinates.

\subsection{Gramians as an outer product}

For the time being, we still work in Floquet coordinates $\vz(t)$, and we begin by showing that the Gramians $\mP(t)$ and $\mQ(t)$ may be computed at specific times $t$ as an outer product of frequency-domain factors.
This is convenient from a computational standpoint, since integrals written as outer products lend themselves to straightforward numerical quadrature.
The content of this subsection may therefore be understood as a first step towards developing an algorithm to compute the Gramians using frequency-domain variables.

\begin{proposition}
\label{prop:grams_as_outer_product}
The $T$-periodic reachability and observability Gramians $\mP(t)$ and $\mQ(t)$ can be written as
\begin{equation}
    \mP(t) = \frac{1}{2\pi}\int_{-\infty}^{\infty}\mZ(\gamma,t)\mZ(\gamma,t)^*\,\mathrm{d}\gamma, \quad
    \mQ(t) = \frac{1}{2\pi}\int_{-\infty}^{\infty}\mY(\gamma,t)\mY(\gamma,t)^*\,\mathrm{d}\gamma,\label{eq:mP_mQ_int}
\end{equation}
where
\begin{align}
    \mZ(\gamma,t) &= \sum_{m\in\mathbb{
    Z}}\mZ_m(\gamma) e^{i m\omega t} = \sum_{m\in\mathbb{Z}}\left(\left(i\gamma\mI - (-i m\omega\mI + \mJ)\right)^{-1} \widetilde\mB_{m}\right)e^{i m \omega t}\label{eq:mZ} \\
    \mY(\gamma,t) &= \sum_{m\in\mathbb{Z}}\mY_m(\gamma)e^{i m\omega t} = \sum_{m\in\mathbb{Z}}\left(\left(-i\gamma\mI - (i m\omega\mI + \mJ^*)\right)^{-1} \widetilde\mC_{-m}^*\right)e^{i m \omega t}.\label{eq:mY}
\end{align}
\end{proposition}
\begin{proof}
The proof relies on the linearity of the Sylvester equations \eqref{eq:sylv_fwd_unstable} and \eqref{eq:sylv_adj_unstable}, so that it can be shown that the Fourier coefficients $\mP_k$ and $\mQ_k$ of the Gramians can be written as linear combination of quantities that satisfy Sylvester equations similar to \eqref{eq:sylv_fwd_unstable} and \eqref{eq:sylv_adj_unstable}.
Details can be found in  \ref{app_subsec:proof_gram_as_outer_products}.
\end{proof}

We have therefore written the reachability Gramian $\mP(t)$ as an outer product of a matrix-valued function $\mZ(\gamma,t)$ and its complex conjugate transpose.
It is now easy to see that for any fixed time $t \in [0,T)$, one can estimate $\mP(t)$ by numerically evaluating the integral in \eqref{eq:mP_mQ_int} via quadrature.
The most computationally-intensive part of evaluating this integral is the
computation of $\mZ_m(\gamma)$.  Remarkably, however, this computation only
needs to be performed once: one may compute and store $\mZ_m(\gamma)$ for every
$m$ and $\gamma$, and then, for every desired time $t \in [0,T)$, we simply have to rotate $\mZ_m(\gamma)$ by~$e^{i m \omega t}$ and evaluate the integral.

\subsection{Gramians in physical coordinates and connection with the harmonic resolvent}

We are now ready to transition to physical coordinates $\vx(t)$, which are well-suited for computation.
Given the Gramians $\mP(t)$ and $\mQ(t)$ in the Floquet coordinates $\vz(t)$, one may verify that the Gramians in the original $\vx(t)$ coordinates are given by
\begin{equation}
    \mP_\vx(t) = \mV(t)\mP(t)\mV(t)^*,\quad \mQ_\vx(t) = \mW(t)\mQ(t)\mW(t)^*,
\end{equation}
where, as before, the $T$-periodic matrices $\mV(t)$ and $\mW(t)$ define the Floquet change of coordinates in \eqref{eq:sys_ltp_lti}.
Using the integral representation of $\mP(t)$ and $\mQ(t)$ in Proposition \ref{prop:grams_as_outer_product}, it readily follows that
\begin{equation}
\label{eq:mPx_mQx_integrals}
    \mP_\vx(t) = \frac{1}{2\pi}\int_{-\infty}^{\infty}\mZ_{\vx}(\gamma,t)\mZ_{\vx}(\gamma,t)^*\,\mathrm{d}\gamma,\,\,
    \mQ_\vx(t) = \frac{1}{2\pi}\int_{-\infty}^{\infty}\mY_{\vx}(\gamma,t)\mY_{\vx}(\gamma,t)^*\,\mathrm{d}\gamma,
\end{equation}
where
\begin{align}
    \mZ_{\vx}(\gamma,t) &\coloneqq \mV(t)\mZ(\gamma,t) = \sum_{k,m\in\mathbb{Z}} \underbrace{\mV_{k-m}\mZ_m(\gamma)}_{\coloneqq\mZ_{\vx,k}(\gamma)}e^{i k\omega t},\label{eq:mZvx}\\
    \mY_{\vx}(\gamma,t) &\coloneqq \mW(t)\mY(\gamma,t) = \sum_{k,m\in\mathbb{Z}} \underbrace{\mW_{k-m}\mY_m(\gamma)}_{\coloneqq\mY_{\vx,k}(\gamma)}e^{i k\omega t}.
\end{align}
Using the definition of $\mZ_m(\gamma)$ and $\mY_m(\gamma)$ in Proposition \ref{prop:grams_as_outer_product}, and the definition of $\widetilde\mB(t)$ and $\widetilde\mC(t)$ in \eqref{eq:sys_ltp_lti}, we may further expand $\mZ_{\vx,k}(\gamma)$ and $\mY_{\vx,k}(\gamma)$ as follows
\begin{align}
    \mZ_{\vx,k}(\gamma) &= \sum_{m\in\mathbb{Z}}\mV_{k-m}\left(i\gamma\mI - (-i m \omega\mI + \mJ)\right)^{-1} \sum_{j\in\mathbb{Z}}\mW_{j-m}^*\mB_j\label{eq:mZx}\\
    \mY_{\vx,k}(\gamma) &= \sum_{m\in\mathbb{Z}}\mW_{k-m}\left(-i\gamma\mI - (i m\omega\mI + \mJ^*)\right)^{-1} \sum_{j\in\mathbb{Z}}\mV_{j-m}^*\mC_{-j}^*.\label{eq:mYx}
\end{align}
We now show that $\mZ_{\vx,k}(\gamma)$ and $\mY_{\vx,k}(\gamma)$ can be computed using the harmonic resolvent operator.

We begin with a short derivation of the harmonic resolvent operator.
More details may be found in \cite{padovan2022}.
Starting from \eqref{eq:sys_ltp} we write the state vector $\vx(t)$ as
\begin{equation}
\label{eq:emp}
    \vx(t) = e^{i\gamma t}\sum_{k\in\mathbb{Z}}\vx_{k+\gamma} e^{i k \omega t},\quad \gamma \in [0,\omega/2].
\end{equation}
The signal above is known as an exponentially modulated periodic (EMP) signal, where the $T$-periodic component inside the sum is modulated by the complex exponential $e^{i\gamma t}$.
It is well-known that EMPs are the appropriate class of signals for the analysis of time-periodic signals (see, e.g., \cite{johnson} or \cite{Wereley91}).
For later reference, we observe that the signal $\vx(t)$ is a sum of
Fourier modes with frequencies in the set $\Omega_\gamma = \gamma + \omega\mathbb{Z}$, where $+$ denotes element-wise addition.
Since all the linear operators in \eqref{eq:sys_ltp} are periodic with period $T$, they may be written in a Fourier series analogous to \eqref{eq:fourier_mP_mQ}.
Then, writing $\vu(t)$ as an EMP, formula \eqref{eq:sys_ltp} may be written in the frequency domain as
\begin{equation}
\label{eq:freq_domain_sys}
    \begin{aligned}
    i (\gamma + k\omega)\vx_{k+\gamma} &= \sum_{j\in\mathbb{Z}}\mA_{k-j}\vx_{j+\gamma} + \sum_{j\in\mathbb{Z}}\mB_{k-j}\vu_{j+\gamma} \\
    \vy_{k+\gamma} &= \sum_{j\in\mathbb{Z}} \mC_{k-j}\vx_{j+\gamma}.
    \end{aligned}
\end{equation}
Letting $\hat\vx_\gamma = \left(\ldots,\vx_{-1+\gamma},\vx_{\gamma},\vx_{1+\gamma},\ldots\right)$ denote an infinite-dimensional vector that contains all the coefficients of the EMP
signal \eqref{eq:emp}, we can define the infinite-dimensional linear operator $\mT$ as
\begin{equation}
\label{eq:mT}
    \left[\mT\hat\vx_\gamma\right]_{k} = - i k\omega \vx_{k+\gamma} + \sum_{j\in\mathbb{Z}}\mA_{k-j}\vx_{j+\gamma}.
\end{equation}
Notice that $\mT$ is independent of $\gamma$, as it depends only on $\omega$ and on the Fourier coefficients of $\mA(t)$.
From this definition and from formula \eqref{eq:freq_domain_sys}, it follows that $\vx_{k+\gamma}$ is given by
\begin{equation}
\label{eq:vx}
    \vx_{k+\gamma} = \sum_{j,l\in\mathbb{Z}}\underbrace{\left[\left(i\gamma \mI - \mT\right)^{-1}\right]_{k,j}}_{\mH_{k,j}(\gamma)}\mB_{j-l}\vu_{l+\gamma}
\end{equation}
where the operator
\begin{equation}
\label{eq:hr_gamma}
    \mH(\gamma) = \left(i\gamma \mI - \mT\right)^{-1}
\end{equation}
is known as the harmonic resolvent operator evaluated at $\gamma$, and $\mH_{k,j}(\gamma)$ is the block of $\mH(\gamma)$ that maps inputs at frequency $(\gamma + j\omega)\in\Omega_\gamma$ to outputs at frequency $(\gamma + k\omega)\in\Omega_\gamma$.
We can now state the desired result.

\begin{proposition}
\label{prop:gram_factors_hr}
The Fourier coefficients $\mZ_{\vx,k}(\gamma)$ and $\mY_{\vx,k}(\gamma)$ of the Gramian factors in formulas \eqref{eq:mZx} and \eqref{eq:mYx} may be written in terms of the harmonic resolvent $\mH(\gamma)$ as follows
\begin{equation}
    \mZ_{\vx,k}(\gamma) = \sum_{j\in\mathbb{Z}}\mH_{k,j}(\gamma)\mB_j,\quad \mY_{\vx,k}(\gamma) = \sum_{j\in\mathbb{Z}}[\mH(\gamma)^*]_{k,j}[\mC^*]_j.
\end{equation}
\end{proposition}
\begin{proof}
Using a Floquet change of coordinates, it can be shown that
\begin{equation}
    \mH_{k,j}(\gamma) = \sum_{m\in\mathbb{Z}}\mV_{k-m}\left(i\gamma\mI -(-i m\omega\mI +\mJ)\right)^{-1}\mW^*_{j-m},
\end{equation}
and this concludes the proof. More details can be found in \ref{app_subsec:proof_gram_factors_hr}.
\end{proof}

Formula \eqref{eq:mPx_mQx_integrals} and Proposition \ref{prop:gram_factors_hr} give us the necessary building blocks for a practical algorithm to compute the Gramians at desired times $t\in[0,T)$.
In particular, the proposition shows us how to evaluate the integrals in \eqref{eq:mPx_mQx_integrals} using the easily accessible harmonic resolvent operator.

\section{An algorithm to compute a factorization of the Gramians}
\label{sec:algorithm_for_gramians}

In this section we highlight some of the features of the harmonic resolvent~$\mH(\gamma)$ that can be exploited to minimize computational cost.
We then provide an explicit algorithm to compute the factors $\mZ_{\vx}(\gamma,t)$ and $\mY_{\vx}(\gamma,t)$.

\subsection{Symmetries in the harmonic resolvent operator}

We henceforth focus on the reachability Gramian $\mP(t)$, since the computation of the observability Gramian $\mQ(t)$ can be carried out in a similar fashion.
For a given time $t\in[0,T)$, we can estimate the integral \eqref{eq:mPx_mQx_integrals} as follows
\begin{equation}
\label{eq:mP_quadrature}
    \mP_\vx(t) = \frac{1}{2\pi}\int_{-\infty}^{\infty}\mZ_\vx(\gamma,t)\mZ_\vx(\gamma,t)^*\mathrm{d}\gamma \approx \frac{1}{2\pi}\sum_{i \in \mathbb{N}} \xi_i \mZ_{\vx}(\gamma_i,t)\mZ_{\vx}(\gamma_i,t)^*
\end{equation}
where $\xi_i$ are quadrature coefficients, $\gamma_i$ are discrete samples over the interval $(-\infty,\infty)$ and
\begin{equation}
    \mZ_{\vx}(\gamma_i,t) = \sum_{k\in\mathbb{Z}}\mZ_{\vx,k}(\gamma_i)e^{i k\omega t} =  \sum_{k\in\mathbb{Z}}\left(\sum_{j\in\mathbb{Z}}\mH_{k,j}(\gamma_i)\mB_j\right) e^{i k \omega t}\in\mathbb{C}^{N\times M}
\end{equation}
by Proposition \ref{prop:gram_factors_hr}.
In order to evaluate $\mZ_\vx(\gamma,t)$ at different times, we can simply
compute and store the Fourier coefficients $\mZ_{\vx,k}(\gamma)$ and then rotate them using the complex exponential.
This is advantageous, since computing $\mZ_{\vx,k}(\gamma)$ is an expensive operation that could easily become computationally intractable if it had to be performed multiple times.
In particular, computing matrix-matrix products of the form $\mH_{k,j}(\gamma)\mB_j$ requires inverting the operator $i\gamma \mI - \mT$ (see formula \eqref{eq:vx}).
The cost of solving these linear systems is dominated either by the computation of a complete factorization of $i \gamma \mI - \mT$ (e.g., LU decomposition), or by the computation of a preconditioner to assist the convergence of iterative solvers such as GMRES.
This cost could be intractable if it had to be sustained for many values $\gamma$ in the interval $(-\infty,\infty)$.
Fortunately, we now show that the factorization (or computation of a
preconditioner) needs to be performed for only a few values of $\gamma$ in the interval $[0,\omega/2]$.


We begin with the following proposition, which states that $\mZ(\alpha,t)$ for any $\alpha \in (-\infty,\infty)$ may be computed using the harmonic resolvent $\mH(\gamma)$ evaluated at $\gamma \in (-\omega/2,\omega/2]$.

\begin{proposition}
\label{prop:alpha_gamma}
For any $\alpha \in \mathbb{R}$, there exists an integer $m$ such that $\gamma = \alpha - m \omega \in (-\omega/2,\omega/2]$ and
\begin{equation}
    \mZ_{\vx}(\alpha,t) = \sum_{k,j\in\mathbb{Z}}\mH_{k,j}(\gamma)\mB_{j-m}e^{i (k-m) \omega t}.
\end{equation}
\end{proposition}
\begin{proof}
Here we present an intuitive reason why this result holds.
Recall that $\mH(\alpha)$ maps inputs over the frequency set $\Omega_\alpha = \alpha + \omega\mathbb{Z}$ to outputs over the same set $\Omega_\alpha$.
Clearly, if $\alpha =\gamma + m\omega$ (for an integer $m$), then $\Omega_\alpha = \Omega_\gamma = \gamma + \omega\mathbb{Z}$.
So, in order to compute $\mZ_{\vx}(\alpha,t)$, we can use the harmonic resolvent $\mH(\gamma)$ evaluated at $\gamma$.
The rigorous proof is in \ref{app_subsec:proof_alpha_gamma}.
\end{proof}

A second observation that we can make to reduce the computational cost stems from the real-valued nature of the dynamics in \eqref{eq:sys_ltp}.
In particular, it can be shown that for every $\gamma$, we have $\mZ(-\gamma,t)=\overline{\mZ(\gamma,t)}$, where the overline denotes complex conjugation. Thus, the desired Gramian $\mP(t)$ may be approximated as
\begin{equation}
\label{eq:mPx_approx}
    \mP_{\vx}(t) \approx \frac{1}{\pi}\sum_{i \in \mathbb{N}} \xi_i \left[ \mZ_{\vx,\text{r}}(\alpha_i,t)\mZ_{\vx,\text{r}}(\alpha_i,t)^* + \mZ_{\vx,\text{i}}(\alpha_i,t)\mZ_{\vx,\text{i}}(\alpha_i,t)^*\right]c_i,
\end{equation}
where $\alpha_i \geq 0$, the subscripts ``$\text{r}$'' and ``$\text{i}$'' denote the real and
imaginary parts of $\mZ_{\vx}(\alpha_i,t)$, and $c_i=1$ for $\alpha_i >0$ and 1/2 for $\alpha_i = 0$.
In other words, the integral can be approximated by considering only positive values $\alpha\in[0,\infty)$.
Putting together this observation and the result from Proposition \ref{prop:alpha_gamma}, it is clear that we only need to factorize (or compute a preconditioner for) $\mH(\gamma)$ at values $\gamma \in [0,\infty)\cap (-\omega/2,\omega/2] = [0,\omega/2]$.

\subsection{Practical algorithm to compute the Gramian factors}

As a first step for practical implementation, we need to truncate the Fourier representations of the periodic components of the dynamics.
In particular, we truncate $\mA(t)$, $\mB(t)$, and $\mC(t)$ at frequency $r_b$, so that, for instance, we write
\begin{equation}
\label{eq:mA_fourier}
    \mA(t) = \sum_{k=-r_b}^{r_b}\mA_k e ^{i k \omega t}.
\end{equation}
Similarly, we truncate the EMP signal in \eqref{eq:emp} as follows
\begin{equation}
\label{eq:vx_emp}
    \vx(t) = e^{i\gamma t}\sum_{k=-r}^r\vx_k e^{i k \omega t},
\end{equation}
where we take $r \geq r_b$.
It follows that $\mT$ is a square matrix with size $(2 r + 1)N$ and structure shown below,
\begin{equation}
\label{eq:bmat}
    \mT = \begin{bmatrix}
    \ddots & \ddots & \ddots & \ddots \\
    \ddots & \hat\mR_{-2} & \hat\mA_{-1} & \hat\mA_{-2} & \ddots \\
    \ddots & \hat\mA_{1} & \hat\mR_{-1} & \hat\mA_{-1} & \hat\mA_{-2} & \ddots \\
    \ddots & \hat\mA_{2} & \hat\mA_{1} & \hat\mR_{0} & \hat\mA_{-1} & \hat\mA_{-2} &\ddots \\
     & \ddots & \hat\mA_{2} & \hat\mA_{1} & \hat\mR_{1} & \hat\mA_{-1} &\ddots \\
     &  & \ddots & \hat\mA_{2} & \hat\mA_{1} & \hat\mR_{2} & \ddots\\
     & & & \ddots & \ddots & \ddots & \ddots \\
    \end{bmatrix},
\end{equation}
where $\hat\mR_{k} = \big({- i k\omega} \mI + \hat\mA_0\big) \in \mathbb{C}^{N\times N}$.
Recalling that $i\gamma \mI-\mT$ acts on vectors~$\hat\vx_\gamma$, we henceforth use the notation $[\hat\vx_\gamma]_k$ to indicate the portion of the vector~$\hat\vx_\gamma$ that gets multiplied by $k$th block-column of $\mT$.
(For clarity, the $k$th block-column is the one containing the block $\hat\mR_{k}$.)
We can now observe that for fixed $m$ and $\gamma$, the quantity $\mZ_{\vx,k}(\alpha) \coloneqq \sum_{j}\mH_{k,j}(\gamma)\mB_{j-m}$ (with $\alpha = \gamma - m\omega$, see Proposition \ref{prop:alpha_gamma}) may be computed simultaneously for all $k \in\{-r,\ldots,r\}$.
For example, taking $\gamma = 0$ and $m = 1$, the quantities $\mZ_{\vx,k}(\omega) = \sum_{j}\mH_{k,j}(0)\mB_{j-1}$ satisfy
\begin{equation}
    \underbrace{-\xbracketMatrixstack{
      \ddots & \ddots & & & \\
      \ddots & \mR_{-1} & \mA_{-1} & & \\
      & \mA_{1} & \mR_{0} & \mA_{-1} & & \\
      & & \mA_{1} & \mR_{1} & \ddots & \\
      & & & \ddots & \ddots
    }}_{i\gamma \mI - \mT = - \mT}
    \underbrace{\xbracketVectorstack{
      \vdots \\ \mZ_{\vx,-1}(\omega) \\ \mZ_{\vx,0}(\omega)  \\ \mZ_{\vx,1}(\omega)  \\ \vdots
    }}_{\hat\mZ_{\vx,\gamma}^{(m)}=\hat\mZ_{\vx,0}^{(1)}} =
    \underbrace{\xbracketVectorstack{
      \vdots \\ \mB_{-2} \\ \mB_{-1} \\ \mB_{0} \\ \vdots
    }}_{\hat\mB^{(m)}= \hat\mB^{(1)}}.
\end{equation}
Notice that the matrix $\hat\mB^{(m)}$ is defined such that $[\hat\mB^{(m)}]_k = \mB_{k-m}$.
In general then, given fixed $m$ and $\gamma$, we compute $\mZ_{\vx,k}(\alpha) = \sum_j\mH_{k,j}(\gamma)\mB_{j-m}$ as follows:
\begin{equation}
    \text{solve}\quad (i\gamma \mI - \mT)\hat\mZ^{(m)}_{\vx,\gamma} = \hat\mB^{(m)},\quad \text{extract}\quad \mZ_{\vx,k}(\alpha) \coloneqq [\hat\mZ^{(m)}_{\vx,\gamma}]_{k}.
\end{equation}

\begin{algorithm}
\caption{Compute factor $\widetilde\mZ_{\vx}(t_n)$ for $t_n \in [0,T)$}\label{alg:compute_factor}
\begin{algorithmic}[1]
\Require Matrix $\mT$ and discrete points $\gamma_i \in [0,\omega/2]$ with $i \in \{1,2,\ldots,L\}$
\Ensure Matrix $\widetilde\mZ_{\vx}(t)$ at time $t\in[0,T)$
\Algphase{Part I: Compute $\mZ_{\vx,k}(\gamma_i + m\omega)\coloneqq \sum_{j}\mH_{k,j}(\gamma)\mB_{j-m}$ for all $\gamma_i$ and $m$}
\State Initialize matrix $\mX \in \mathbb{C}^{(2 r +1)N\times ((L-2)(2 r+1) + 2(r+1))M}$
\For{$i \in \{1,2,\ldots,L\}$}
    \State Compute factorization (or preconditioner) of $i \gamma_i \mI - \mT$
    \If{$\gamma_i \neq 0$ and $\gamma_i \neq \omega/2$}
    \State Range $=\{-r,\ldots,0,\ldots,r\}$
    \Else
    \State Range $=\{0,\ldots,r\}$
    \EndIf
    \For {$m \in$ Range}
        \State Solve $(i\gamma_i \mI - \mT)\hat\mZ_{\vx,\gamma_i}^{(m)} = \hat\mB^{(m)}$, where $[\hat\mZ_{\vx,\gamma_i}^{(m)}]_k = \mZ_{\vx,k}(\gamma_i+m\omega)$
        \State Store $\hat\mZ^{(m)}_{\vx,\gamma_i}$ in $\mX$
    \EndFor
\EndFor
\State \textbf{Return:} Matrix $\mX$
\Algphase{Part II: Compute $\widetilde\mZ_\vx(t_n)$ at some desired time $t_n\in[0,T)$}

\State Initialize $\widetilde\mZ_{\vx}(t_n) \in \mathbb{R}^{N \times 2 ((L-2)(2 r+1) + 2(r+1))M}$
\For{$i \in \{1,2,\ldots,L\}$}
    \If{$\gamma_i \neq 0$ and $\gamma_i \neq \omega/2$}
    \State Range $=\{-r,\ldots,0,\ldots,r\}$
    \Else
    \State Range $=\{0,\ldots,r\}$
    \EndIf
    \If{$\gamma_i = 0$}
    \State $c_i = 1/2$
    \Else
    \State $c_i = 1$
    \EndIf
    \For {$m \in$ Range}
        \State Extract the component $\hat\mZ_{\gamma_i}^{(m)}$ from $\mX$
        \State Compute $\mZ_{\vx,\gamma_i}^{(m)}(t_n)\leftarrow \sum_{k\in\mathbb{Z}_X}[\hat\mZ_{\gamma_i}^{(m)}]_k e^{i (k-m)\omega t_n}$
        \State Store $\sqrt{\frac{c_i\xi_i}{\pi}}\text{Real}(\mZ_{\vx,\gamma_i}^{(m)}(t_n))$ and $\sqrt{\frac{c_i\xi_i}{\pi}}\text{Imag}(\mZ_{\vx,\gamma_i}^{(m)}(t_n))$ into $\widetilde\mZ_{\vx}(t_n)$
    \EndFor
\EndFor
\State \textbf{Return:} Factor $\widetilde\mZ_{\vx}(t_n)$
\end{algorithmic}
\end{algorithm}

We are now ready to present an algorithm to compute the Gramian factors.
In particular, Algorithm \ref{alg:compute_factor} will output a matrix $\widetilde\mZ_{\vx}(t)$ with columns
\begin{equation}
    \frac{1}{\sqrt{\pi}} \left\{\sqrt{c_i\xi_i}\mZ_{\vx,\text{r}}(\alpha_i,t),\sqrt{c_i\xi_i}\mZ_{\vx,\text{i}}(\alpha_i,t)\right\},\quad \alpha_i \in [0,\infty),
\end{equation}
so that, per equation \eqref{eq:mPx_approx}, we have $\mP_\vx(t) \approx \widetilde\mZ_{\vx}(t)\widetilde\mZ_{\vx}(t)^*$.

In the first part of the algorithm, we compute the frequency-domain factors and store them.
This is the most computationally intensive part of the algorithm, as we need to solve several linear systems of size $(2r+1)N$.
In the second part, we simply rotate the previously computed factors using the complex exponential, and evaluate the factor $\mZ_{\vx}(t)$ at the desired time $t\in [0,T)$.
This part of the algorithm is virtually free of cost
compared to the first part.
We close this section by observing that the factors $\widetilde\mY_\vx(t)$ may be computed using Algorithm \ref{alg:compute_factor} by replacing $i\gamma_i \mI - \mT$ with its complex conjugate transpose and by replacing $\mB$ with $\mC^*$.
The observability Gramian at any time $t$ may then be evaluated via quadrature as $\mQ_\vx(t)\approx\widetilde\mY_{\vx}(t)\widetilde\mY_{\vx}(t)^*$.

\subsection{Computational considerations}
\label{subsec:computational_considerations}

We now discuss the computational advantages and drawbacks of our proposed approach to compute the Gramian factors in the frequency domain.
Once again, we consider the reachability Gramian $\mP_{\vx}(t)$, since analogous logic applies to the observability Gramian $\mQ_{\vx}(t)$.
In order to appreciate the benefits and drawbacks of using frequential Gramians, it is instructive to understand how the Gramians would be computed in the time domain.
For a \emph{stable} system \eqref{eq:sys_ltp}, we recall that
\begin{equation}
\label{eq:mPx_int_time}
    \mP_{\vx}(t) = \mP_{\vx}(t+T) \coloneqq \lim_{n\rightarrow\infty}\mG_{R,\vx}(0,n T + t),
\end{equation}
where $n$ is a positive integer and
\begin{equation}
    \mG_{R,\vx}(0,s) = \int_{0}^{s}\mF(s,\tau)\mB(\tau)\mB(\tau)^*\mF(s,\tau)^*\mathrm{d}\tau.
\end{equation}
Here, $\mF(s,\tau)$ denotes the fundamental solution of \eqref{eq:sys_ltp};
i.e., $\mathrm{d}/\mathrm{d}s (\mF(s,\tau)) = \mA(s) \mF(s,\tau)$ with $\mF(\tau,\tau)=\mI$.
For a fixed $\tau \leq s$, the quantity
\begin{equation}
    \vx_{\tau}(s) \coloneqq \mF(s,\tau)\mB(\tau) \in \mathbb{R}^{N\times M}
\end{equation}
is the time-$\tau$ impulse response of \eqref{eq:sys_ltp}, which can be computed numerically by solving
\begin{equation}
\label{eq:time_step}
    \frac{\mathrm{d}}{\mathrm{d} t}\vx^{(j)}_{\tau}(t) = \mA(t)\vx^{(j)}_{\tau}(t),\quad \vx^{(j)}_{\tau}(\tau) = \mB_j(\tau),\quad j \in \{1,2,\ldots,M\},
\end{equation}
from $t=\tau$ to $t = s$.
Formula \eqref{eq:mPx_int_time} naturally lends itself to numerical quadrature, so that $\mP_{\vx}(t)$ may be approximated as $\mP_{\vx}(t) \approx\lim_{n\rightarrow\infty} \mX(n T + t)\mX(n T + t)^*$, where $\mX(s)$ is given below
\begin{equation}
    \mX(s) =\bigg[\sqrt{\xi_1} \vx_{\tau_1}(s),\sqrt{\xi_2} \vx_{\tau_2}(s),\ldots,\sqrt{\xi_S}\vx_{\tau_S}(s)\bigg]\in\mathbb{C}^{N\times M S},\quad \tau_i \in [0,s],
\end{equation}
and $\xi_i$ are quadrature coefficients.
By contrast, we recall from part II of Algorithm \ref{alg:compute_factor} that the Gramian $\mP_{\vx}(t)$ may be approximated in the frequency domain as $\mP_{\vx}(t)\approx \widetilde\mZ_{\vx}(t)\widetilde\mZ_{\vx}(t)^*$, where $\widetilde\mZ_{\vx}(t)$ has size $N\times 2 M I$ and $I$ is equal to the number of discrete $\alpha_i$ in equation \eqref{eq:mPx_approx}.
The factor of $2$ comes from the fact that the real and imaginary parts of $\mZ_{\vx}(\alpha,t)$ in equation \eqref{eq:mPx_approx} are stored separately.

With this information at hand, we see that evaluating $\widetilde\mZ_{\vx}(t)$ at some fixed time $t$ requires solving $M I$ algebraic systems of equations.
By contrast, evaluating $\mX(s)$ at some fixed time $s = n T + t$ requires solving $M S$ initial-value problems \eqref{eq:time_step} in the time domain.
The question we ask is, when is it convenient to compute $\mP_{\vx}(t)$ in the frequency domain, and when is it convenient to compute it in the time domain?
Although it is virtually impossible to provide a precise operations count for the two methods, we can still provide guidelines that the user may find useful.
For simplicity, let us assume that $O(I) = O(S)$. This can usually be taken to be the case in practice.
The fundamental difference between the two methods is that one requires time stepping, while the other does not.
Thus, although the number of required impulse responses $M S$ and required
linear solves $M I$ is comparable, the time-domain method is fundamentally
limited by \textit{(i)} numerical stability constraints associated with time
stepping (i.e., the time step in \eqref{eq:time_step} might have to be small,
depending on the specific properties of the time-stepper) and $\textit{(ii)}$
physical transients that time-steppers are forced to follow.
In particular, if the dynamics exhibit slowly-decaying transients, then $s = n T+t$ must be taken very large (see formula \eqref{eq:mPx_int_time}), and the cost of computing the Gramian using the time domain increases significantly.
By constrast, computing the Gramian factors in the frequency domain requires
solving algebraic systems of equations that do not suffer from the drawback of
time-stepping methods.
Another often overlooked aspect of time steppers is that they are inherently sequential in time, so that any type of computational parallelism can only be spatial (e.g., distributing the degrees of freedom of \eqref{eq:sys_ltp} across multiple processors).
On the other hand, algebraic systems of equations lend themselves to massive space-time parallelism, so that a higher number of processors can be deployed to accelerate the computations in the frequency domain.

Unfortunately, the benefits of the frequency domain are not free of cost.
The main drawback is that the $N$-dimensional time-periodic dynamics are ``lifted" into a higher-dimensional space of size $N(2 r + 1)$.
Therefore, solving the desired algebraic equations requires inverting the operator $i\gamma \mI -\mT$.
The main issue therefore lies in computing a factorization of $i\gamma - \mT$, or a suitable preconditioner (although, fortunately, this only has to be done for a few values $\gamma \in [0,\omega/2)$, as explained in the previous section).
Depending on the nature and size of the underlying problem, it may be possible
to compute an LU decomposition of $i\gamma \mI -\mT$ using parallelized
libraries such as MUMPS \citep{mumps}, or alternatively, one may use out-of-the-box preconditioners.
Given the structure of $\mT$, we recommend the use of Block-Jacobi as a starting point.
However, there may be problems whose size is such that an LU factorization cannot be performed, and whose structure is such that a Block-Jacobi preconditioner does not work particulalry well.
Identifying an efficient preconditioner tailored to the quasi block-Toeplitz structure of the matrix $i\gamma \mI -\mT$ remains an open question and the subject of future work.

As far as storage is concerned, the time-domain approach is more efficient.
Although both approaches yield Gramian factors of comparable size, part I of
Algorithm \ref{alg:compute_factor} shows that the matrix $\mX$ of size $N(2r + 1)\times M I$ needs to be held in memory.
Thus, the memory burden of the frequency-domain approach is $O(r)$ higher than the time-domain method.

Given this discussion, it becomes clear that the choice of algorithm is heavily dependent on the nature of the underlying dynamics.
In general, we recommend the use of the frequency domain for systems with state $N$ of moderate size $O(10^5)$ or less, so that the size of $\mT$ remains below $O(10^7)$ and an LU decomposition of $i\gamma\mI-\mT$ can be computed.
If the size of the system is much higher than $O(10^5)$, but the structure of $\mT$ is predominantly diagonal, then the frequency domain remains a feasible option, with out-of-the-box preconditioners such as Block-Jacobi assisting the convergence of Krylov solvers.
Systems where the structure of $\mT$ is dominantly diagonal (see, e.g., \eqref{eq:bmat}) are systems where $\lVert \mA_{k > 0}\rVert \ll \lVert \mA_0\rVert$, where $\mA_k$ is the $k$th Fourier coefficient of the operator $\mA(t)$.
We also recommend the use of the frequency domain if the underlying system exhibits slowly-decaying transients, since, as explained, these will inevitably drive up the cost of computing the Gramians in the time domain.
Finally, the frequency domain is also well-suited for unstable systems, while the time-domain procedure described herein cannot be applied since the limit in \eqref{eq:mPx_int_time} does not exist.
Theoretically, this issue could be solved via a splitting of the stable and unstable Floquet eigenspaces, but this comes at the price of higher computational cost.

\section{Continuous-time balanced truncation}
\label{sec:bt}

In this section we describe the continuous-time balanced truncation approach for model reduction.
Given a dynamical system of the form \eqref{eq:sys_ltp}, balanced truncation seeks a \emph{continuously differentiable} periodic change of coordinates $\vx(t) = \mPhi(t)\vq(t)$, with $\mPsi(t)^*\mPhi(t) = \mI \in \mathbb{R}^{N\times N}$, such that the $\vq(t)$-coordinate Gramians
\begin{equation}
    \mP_{\vq}(t) = \mPsi(t)^*\mP_{\vx}(t)\mPsi(t),\quad \mQ_{\vq}(t) = \mPhi(t)^*\mQ_{\vx}(t)\mPhi(t)
\end{equation}
are equal and diagonal.
In other words,
\begin{equation}
    \mP_\vq(t) = \mQ_\vq(t) = \mSigma(t) = \text{diag}(\sigma_1(t),\sigma_2(t),\ldots,\sigma_N(t)).
\end{equation}
We begin by illustrating the balancing scheme, and then we address some of the subtleties associated with the time-varying nature of the problem.
Given a factorization of the Gramians, e.g.,
\begin{equation}
    \mP_\vx(t) = \widetilde\mZ_\vx(t)\widetilde\mZ_\vx(t)^*,\quad \mQ_\vx(t) = \widetilde\mY_\vx(t)\widetilde\mY_\vx(t)^*,
\end{equation}
where the factors can be chosen to be $T$-periodic (see, e.g., the factors described in the previous sections),
the first step in computing the balancing change of coordinates is to compute the singular value decomposition (SVD)
\begin{equation}
\label{eq:svd}
    \widetilde\mY_\vx(t)^*\widetilde\mZ_\vx(t) = \mU(t)\mSigma(t)\mV(t)^*
\end{equation}
at all times $t\in [0,T)$.
Then, so long as $\mSigma(t)^{-1}$ exists for all $t$, the desired matrices $\mPhi(t)$ and $\mPsi(t)$ are given by
\begin{equation}
\label{eq:mPhi_mPsi}
    \mPhi(t) = \widetilde\mZ_\vx(t)\mV(t)\mSigma(t)^{-1/2},\quad \mPsi(t) = \widetilde\mY_\vx(t)\mU(t)\mSigma(t)^{-1/2},
\end{equation}
where it can be easily checked that $\mPsi(t)^*\mPhi(t) =\mI$ for all times.

The first subtlety stems from the fact that we  require the matrices $\mPhi(t)$ and $\mPsi(t)$ to be continuously differentiable.
As a consequence, while in linear time-invariant systems the (time-invariant) Hankel singular values $\sigma_i$ may be arranged in descending order, in the time-varying setting it may not be possible to enforce this arrangement for all times.
In fact, the differentiability requirement can cause the time-periodic
$\sigma_i(t)$ to cross each other for different~$i$.
This seemingly innocuous difference may lead to ambiguities when trying to determine a global (i.e., for all times) truncation rank $r \ll N$ to assemble a reduced-order model.
In fact, it is theoretically possible that the smallest singular value at time $t = 0$ becomes the largest at a later time.
Fortunately, while the singular values can certainly coalesce in practice, the crossing is usually localized, so that the aforementioned pathological behavior is uncommon.
For the sake of completeness, we observe that this issue can be addressed by allowing for the reduced-order system to have time-varying dimensions.
While this is a natural thing to do in discrete-time setting \citep{varga2000}, the possibility of time-varying state dimensions in continuous-time has been explored by \cite{sandberg2004} at the price of introducing discontinuities in the measured output $\vy(t)$.

Another subtlety that arises in continuous-time balanced truncation is the fact that even though the full-order model \eqref{eq:sys_ltp} is periodic with period $T$, the balanced model may have period $m T$ for some integer $m \geq 1$.
This is a direct consequence of the fact that smooth decompositions of $T$-periodic matrices (in this case $\widetilde\mY_\vx(t)^*\widetilde\mZ_\vx(t)$) may yield factors (see equation \eqref{eq:svd}) whose period is larger than $T$.
A thorough discussion on smoothness and periodicity of some matrix decompositions may be found in \cite{Chern2001}.
These two difficulties are inherently tied to the balancing procedure and they cannot be avoided.
Fortunately, however, they can be addressed in a straightforward fashion within our framework.
An algorithm is provided below.

\begin{algorithm}
\caption{Compute matrices $\mPhi(t)$ and $\mPsi(t)$}\label{alg:compute_balancing}
\begin{algorithmic}[1]
\Require Matrix $\mT$, discrete samples $\gamma_i \in [0,\omega/2]$, reduced-order model rank $r$, expected period $m T$ of $\mPhi(t)$ and $\mPsi(t)$, discrete time samples $t_n \in [0,m T)$
\Ensure Matrices $\mPhi(t_n),\,\mPsi(t_n)\in\mathbb{R}^{N\times r}$ at time instances $t_n \in [0,m T)$
\State Compute and store matrix $\mX$ using part I of Algorithm \ref{alg:compute_factor}
\For{$t_n \in [0,m T)$}
    \State Compute $\widetilde\mZ_\vx(t_n)$ and $\widetilde\mY_\vx(t_n)$ using part II of Algorithm \ref{alg:compute_factor}
    \State Compute the SVD of $\widetilde\mY_\vx(t_n)^*\widetilde\mZ_\vx(t_n) = \mU(t_n)\mSigma(t_n)\mV(t_n)^*$ as in \eqref{eq:svd}
    \State Truncate the SVD factors at rank $r$
    \If{$t_n > 0$}
        \State Order the SVD factors so that the are continuously differentiable
    \EndIf
    \State Compute $\mPhi(t_n)$ and $\mPsi(t_n)$ as in \eqref{eq:mPhi_mPsi}
\EndFor
\end{algorithmic}
\end{algorithm}

The ``if statement" in the algorithm may be understood as a mode-tracking step. Given a sufficiently finely sampled time interval, it is reasonable to expect the SVD factors $\mU(t_n)$ and $\mV(t_n)$ to be well-aligned with the factors $\mU(t_{n-1})$ and $\mV(t_{n-1})$.
Comparing the mode alignment at neighboring time instances allows us to detect any crossing of the singular values and to keep the factors continuously differentiable across the entire interval $[0,m T)$.
For completeness, it is worth observing that different mode-tracking logic can be implemented.
For example, this could be done using an approach similar to the dynamical low rank approximation described in \cite{othmar2007} and in \cite{lubich2014}.
In Algorithm \ref{alg:compute_balancing}, the time interval $[0,m T)$ itself is given as an input to the algorithm.
Unfortunately, there is no practical a-priori way of determining what the period of the factors will be \citep{Chern2001}, so the appropriate value $m T$ will be problem dependent.
Fortunately, it is inexpensive to try different values of $m$ (or even to choose $m$ sufficiently large and then identify the minimal period), since the computationally-intensive part of algorithm~\ref{alg:compute_balancing} consists in computing the matrix $\mX$ via part I of Algorithm \ref{alg:compute_factor}.

Given the matrices $\mPhi(t)$ and $\mPsi(t)$ of period $m T$, the desired $m T$-periodic $r$-dimensional reduced model is obtained by substituting $\vx(t) = \mPhi(t)\vq(t)$ into \eqref{eq:sys_ltp} and left-multiplying by $\mPsi(t)^*$,
\begin{equation}
\label{eq:rom}
    \begin{aligned}
    \ddt\vq(t) &= \underbrace{\mPsi(t)^*\left(\mA(t)\mPhi(t)-\ddt\mPhi(t)\right)}_{\mA_r(t)}\vq(t) + \underbrace{\mPsi(t)^*\mB(t)}_{\mB_r(t)}\vu(t) \\
    \vy(t) &= \underbrace{\mC(t)\mPhi(t)}_{\mC_r(t)}\vq(t).
    \end{aligned}
\end{equation}
This is a $m T$-periodic linear system whose size $r \ll N$ is suitable for control and estimation.

\section{Application to an axisymmetric jet}

In this section we demonstrate the balancing algorithms \ref{alg:compute_factor} and \ref{alg:compute_balancing} on a periodically-forced incompressible axisymmetric jet at two different Reynolds numbers, $Re = 1250$ and $Re = 1500$.
At $Re = 1250$, the flow admits a stable time-periodic solution characterized by
unpaired vortex rings.
However, this solution is extremely sensitive to subharmonic perturbations, so
that any small perturbation will cause neighboring vortex rings to pair and merge.
Here, we compute a reduced-order model and we design a disturbance-rejection feedback controller to delay and mitigate the pairing phenomenon.
At $Re = 1500$, the flow admits an \emph{unstable} time-periodic solution, also
characterized by unpaired vortex rings.
Given the unstable nature of the solution, however, the flow will naturally depart from the unstable unpaired configuration and it will settle onto a different periodic orbit characterized by paired rings.
In this case, we compute a reduced-order model and we use it to design a stabilizing feedback controller.
These two cases demonstrate the effectiveness of algorithms \ref{alg:compute_factor} and \ref{alg:compute_balancing} at delivering a reduced-order model both when the underlying dynamics are stable and when they are unstable.

\subsection{Flow description and numerical setup}

We begin by providing a brief description of the governing equations.
Throughout, velocities are non-dimensionalized by the jet centerline velocity~$U_0$ and lengths are non-dimensionalized by the jet diameter $D_0$, so that we may define the Reynolds number $Re = U_0 D_0/\nu$, where $\nu$ is the kinematic viscosity of the fluid.
The flow is governed by the incompressible Navier-Stokes equation along with the continuity equation over the spatial domain $\mathcal{D} = \{(z,\xi)\lvert \,z\in [0,L_z],\,\xi\in [0,L_\xi]\}$, with $L_z = 15$ and $L_\xi=4$.
In particular, given the (dimensionless) axial velocity $u$, the radial velocity $v$ and the pressure $p$, we have
\begin{subequations}
  \label{eq:ns-all}
\begin{align}
    \frac{\partial u}{\partial t} + u\frac{\partial u}{\partial z} + v\frac{\partial u}{\partial \xi} &= -\frac{\partial p}{\partial z} + \frac{1}{Re}\left(\frac{1}{\xi}\frac{\partial}{\partial \xi}\left(\xi\frac{\partial u}{\partial \xi}\right) + \frac{\partial^2 u}{\partial z^2} \right) \label{eq:ns_z}\\
    \frac{\partial v}{\partial t} + u\frac{\partial v}{\partial z} + v\frac{\partial v}{\partial \xi} &= -\frac{\partial p}{\partial \xi} + \frac{1}{Re}\left(\frac{1}{\xi}\frac{\partial}{\partial \xi}\left(\xi\frac{\partial v}{\partial \xi}\right) - \frac{v}{\xi^2} + \frac{\partial^2 v}{\partial z^2} \right)\label{eq:ns_xi} \\
    \frac{\partial u}{\partial z} + \frac{1}{\xi}\frac{\partial\left(\xi v\right)}{\partial \xi} &= 0.\label{eq:ns_cont}
\end{align}
\end{subequations}
At the centerline $\xi=0$ we impose axisymmetric boundary conditions, at the outflow and at the top boundary we impose a zero normal gradient boundary condition on both velocity components, and at the inflow we consider the axial velocity profile
\begin{equation}
    \label{eq:forcing_profile}
    u(\xi,z=0,t) = g(\xi) \left(1 + A\cos \omega t \right),
  \end{equation}
 where $A$ is the non-dimensional forcing amplitude, $\omega$ is the forcing frequency, and
 \begin{equation}
    g(\xi) = \frac{1}{2}\bigg\{1 - \tanh\left[\frac{1}{4\theta_0}\left(\xi - \frac{1}{4\xi} \right)\right]\bigg\}.
\end{equation}
The parameter $\theta_0$ may be understood as a non-dimensional vorticity thickness of the incoming profile.
The spatial domain is discretized on a fully-staggered grid using second-order finite differences, except for the advective term, which is treated using a third-order upwind-biased scheme.
Given the fully-staggered nature of the grid, we do not require explicit pressure boundary conditions.
Throughout, we work on a grid of size $N_z\times N_\xi = 600\times 200$, we fix $A = 0.05$, $\theta_0 = 0.025$ and $\omega = 2\pi 0.6$, and we consider two different Reynolds numbers $Re = 1250$ and $Re = 1500$.
For this choice of parameters, it is shown by \cite{ardali_jet19} via a Floquet stability analysis that the $T$-periodic solution is stable at $Re = 1250$ and unstable at $Re = 1500$.
For both cases, we compute the $T$-periodic solution via time-stepping of the Navier-Stokes equations \eqref{eq:ns-all} augmented with the time-delay feedback technique described in \cite{ardali_time_delayed_17}.
This technique is necessary to compute unstable solutions via time-stepping (as in the $Re = 1500$ case), but it can also be used to suppress transients in stable configurations (e.g., the $Re = 1250$ case) thereby accelerating the convergence to the desired post-transient solution.
Representative snapshots from the two solutions are shown in figure \ref{fig:bflow}.

\begin{figure}
\centering
\begin{minipage}{0.48\textwidth}
\begin{tikzonimage}[trim= 15 15 35 115,clip,width=0.95\textwidth]{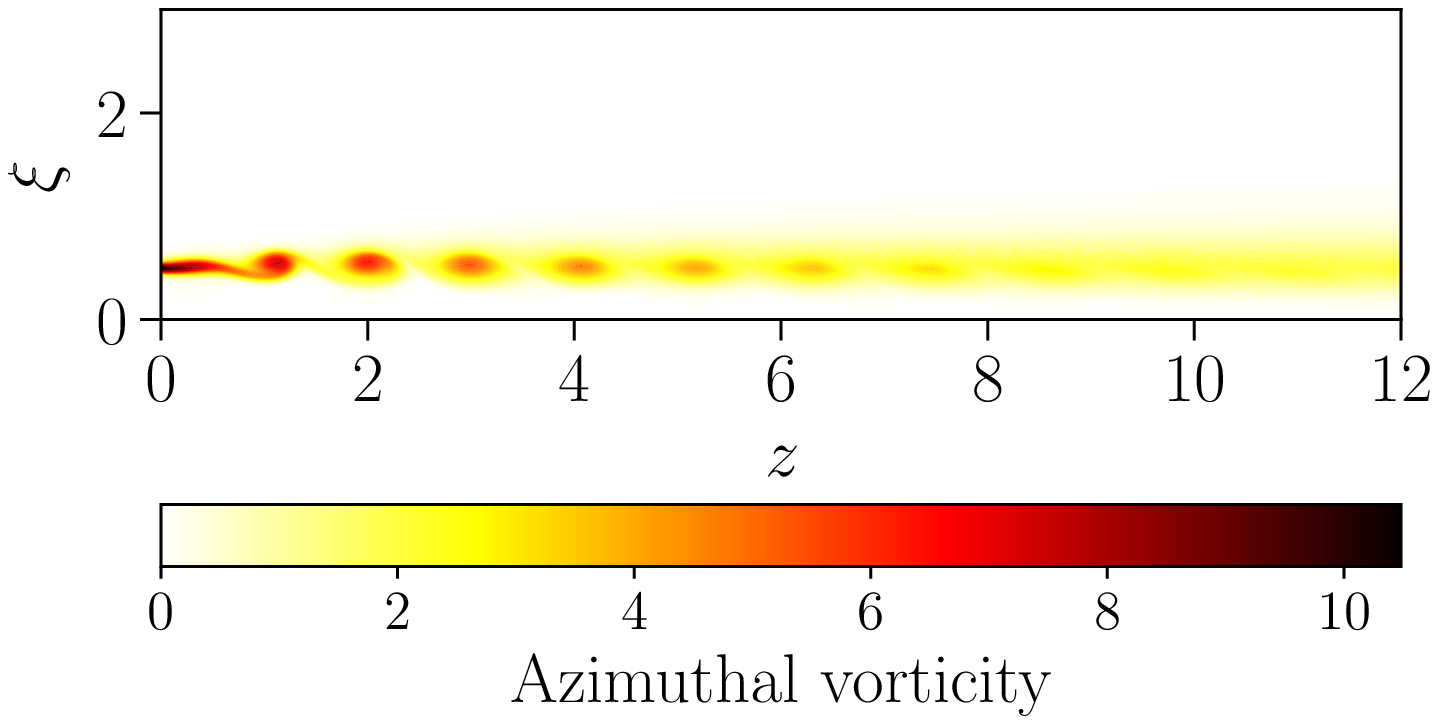}
\node at (0.90,0.85) {\small $\textit{(a)}$};
\end{tikzonimage}
\end{minipage}
\begin{minipage}{0.48\textwidth}
\begin{tikzonimage}[trim= 15 15 35 115,clip,width=0.95\textwidth]{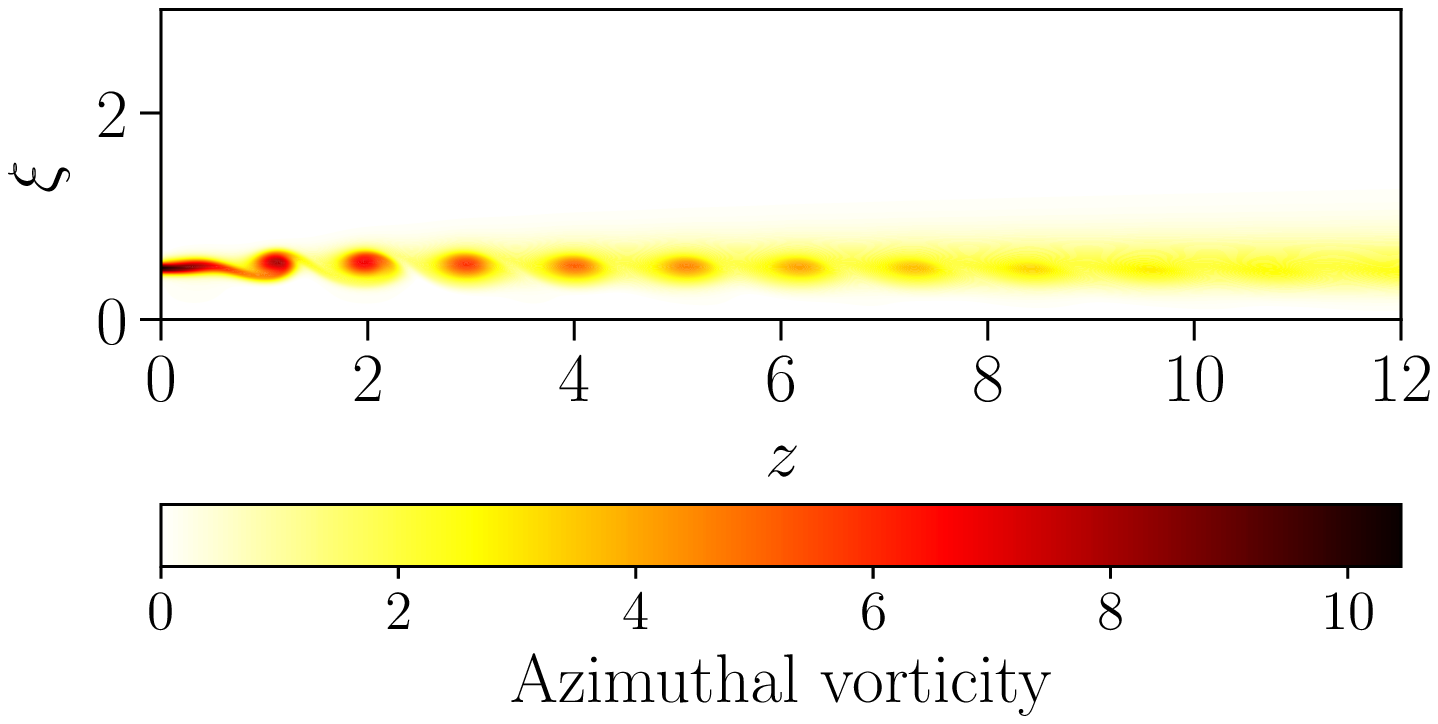}
\node at (0.90,0.85) {\small $\textit{(b)}$};
\end{tikzonimage}
\end{minipage}
\caption{Vorticity snapshots from the $T$-periodic solution at time $t = 0$ for \textit{(a)} Reynolds number $Re = 1250$ and \textit{(b)} Reynolds number $Re = 1500$.}
\label{fig:bflow}
\end{figure}

\subsection{Actuator and sensor configurations}


We now provide details concerning the model-reduction procedure and we also describe the actuator and sensor configurations.
Given the $T$-periodic solution $\vX = \left(U,V\right)$ of the Navier-Stokes equations \eqref{eq:ns-all}, the perturbed velocity field $\vx' = \left(u',v'\right)$ and the perturbed pressure $p'$, the linearized dynamics are governed by the equations

\begin{subequations}
\label{eq:ns_lin}
\begin{align}
    \frac{\partial u'}{\partial t} + B_z(\vx',\vX) &= -\frac{\partial p'}{\partial z} + \frac{1}{Re}\left(\frac{1}{\xi}\frac{\partial}{\partial \xi}\left(\xi\frac{\partial u'}{\partial \xi}\right) + \frac{\partial^2 u'}{\partial z^2} \right) \\
    \frac{\partial v'}{\partial t} + B_\xi(\vx',\vX) &= -\frac{\partial p'}{\partial \xi} + \frac{1}{Re}\left(\frac{1}{\xi}\frac{\partial}{\partial \xi}\left(\xi\frac{\partial v'}{\partial \xi}\right) - \frac{v'}{\xi^2} + \frac{\partial^2 v'}{\partial z^2} \right) \\
    \frac{\partial u'}{\partial z} + \frac{1}{\xi}\frac{\partial\left(\xi v'\right)}{\partial \xi} &= 0,
\end{align}
\end{subequations}
where
\begin{align*}
    B_z(\vx',\vX) &= u'\frac{\partial U}{\partial z} + U\frac{\partial u'}{\partial z} + v'\frac{\partial U}{\partial \xi} + V\frac{\partial u'}{\partial \xi} \\
    B_\xi(\vx',\vX) &= u'\frac{\partial V}{\partial z} + U\frac{\partial v'}{\partial z} + v'\frac{\partial V}{\partial \xi} + V\frac{\partial v'}{\partial \xi}.
\end{align*}
The boundary conditions on the perturbed velocity are analogous to those imposed
on the full velocity field, except for the inflow where we impose zero Dirichlet boundary conditions on both velocity components.
Upon removal of the pressure $p'$ via a Poisson equation and letting $\vx'(t)$ denote the spatially-discretized velocity at the cell faces of the computational grid, the system \eqref{eq:ns_lin} may be written as a linear time-periodic system in standard form
\begin{equation}
\label{eq:ns_ltp_unf}
    \ddt \vx'(t) = \mA(t)\vx'(t),\quad \mA(t) = \mA(t+T).
\end{equation}
Given our grid size, the state vector $\vx'(t)$ will have size $N = 2 N_z N_\xi = 2.4\times 10^5$.

At this point, we need to specify how the control input will enter the dynamics.
This is our first design choice, and we decide to actuate the flow through an
axial velocity body force localized in the proximity of $z_c = 1$ and $\xi_c = 0.5$,
with magnitude
\begin{equation}
\label{eq:mB_control}
    \exp\left[-\frac{1}{\theta_0}\left((z-z_c)^2 +
        (\xi-\xi_c)^2\right)\right] u(t),
\end{equation}
where $u(t)$ is our control input.  Thus, the matrix $\mB$ in equation~\eqref{eq:sys_ltp} is
a single column of height $N$.  Notice that here, the matrix~$\mB$ is time-invariant, which means that the control input always enters the dynamics at the same physical location.
We choose the location $(z_c,\xi_c)$ near the jet nozzle, since this is where one might be able to place an actuator in practice.
Additional insight into the actuator placement was also obtained from the sensitivity analysis in \cite{padovan2022}, where the authors showed that the flow is most sensitive to axial velocity perturbations in the proximity of the nozzle.

The second design choice concerns sensor placement.
Throughout, we choose to measure the axial velocity at four different locations with radial coordinate $\xi_c = 0.5$ and axial coordinates $z_c = \{1.5,2.5,5,6\}$.
This yields a time-invariant matrix $\mC \in \mathbb{R}^{4\times N}$, where each row is given by a spatial profile analogous to the one in \eqref{eq:mB_control}.
In choosing the sensor locations we considered the following.
First, one or more sensors should be placed in proximity of the actuator in order to mitigate the detrimental effect of delays between the input and the measured response.
Second, one or more sensors should be placed farther downstream since this is
the location of the vortex pairing phenomenon that we wish to suppress.
Given our $\mB$ and $\mC$ matrices, we henceforth work with the input-output system below
\begin{equation}
\label{eq:ltp_jet}
    \begin{aligned}
    \ddt \vx'(t) &= \mA(t)\vx'(t) + \mB u(t)\\
    \vy(t) &= \mC\vx'(t),
    \end{aligned}
\end{equation}
where $u(t)\in\mathbb{R}$ is our control input (which will be determined by an appropriate feedback law) and $\vy(t)\in\mathbb{R}^4$ is the measured output.

In order to compute the balancing transformation matrices $\mPhi(t)$ and $\mPsi(t)$ using algorithms \ref{alg:compute_factor} and \ref{alg:compute_balancing}, we need to assemble the matrix $\mT$ \eqref{eq:bmat} associated with the linearized dynamics \eqref{eq:ltp_jet}.
In particular, we truncate the Fourier representation of $\mA(t)$ at $r_b = 4$ harmonics of the fundamental frequency $\omega$ (see formula \eqref{eq:mA_fourier}) and we truncate the EMP representation of the state $\vx'(t)$ at $r = 6$ harmonics (see formula \eqref{eq:vx_emp}).
Thus, the matrix $\mT$ will have size $(2 r +1)N = 3.12\times 10^6$.
Algorithm \ref{alg:compute_factor} is implemented in a PETSc-based solver run on the Princeton Tiger Cluster, and the linear solvers in the algorithm are preconditioned with PETSc's built-in Block-Jocobi preconditioner.
The interval $[0,\omega/2]$ in Algorithm \ref{alg:compute_factor} is discretized using 11 equally-spaced points $\gamma_l\in[0,\omega/2]$.
After computing the balancing transformation matrices $\mPhi(t)$ and $\mPsi(t)$,
we can explicitly assemble a reduced-order model of the form~\eqref{eq:rom},
where in this case the matrices $\mB$ and $\mC$ are time-invariant.

\subsection{Feedback controller and state estimator design}
\label{subsec:controller_observer_design}
Given a reduced-order model of the form \eqref{eq:rom} with reduced state $\vq(t)\in\mathbb{R}^r$, we can now approach the task of designing a feedback controller to modify the full-order dynamics.
We design the feedback controller using the linear quadratic regulator (LQR) approach for linear time-periodic systems.
A thorough overview of the LQR problem for time-periodic systems is given in \cite{Wereley91}, while rigorous results on the existence and uniqueness of a periodic feedback law may be found in \cite{bittanti}.
Simply put, given the $m T$-periodic linear system \eqref{eq:rom}, the LQR method yields a $m T$-periodic feedback control law $u(t)=-\mK(t)\vq(t)$ by solving the optimization problem
\begin{equation}
    \min_{\mK(t)}\quad \mathcal{J}_{LQR} = \int_0^\infty \big(\vq(t)^*\mQ_q(t)\vq(t) + u(t)^2\big)\,\mathrm{d}t,
\end{equation}
subject to the dynamics in \eqref{eq:rom}.
Here, $\mQ_q(t)$ is a positive-semidefinite
$r\times r$ matrix that quantifies the relative importance
of driving the states to zero, versus maintaining small control inputs.
Our choice of $\mQ_q(t)$ is informed by the analysis carried out in section 5C of \cite{padovan2022}.
In particular, for $Re = 1250$, we demonstrated that the pairing phenomenon that we wish to suppress is driven exclusively by a $2 T$-periodic mode denoted~$\vxi(t)$.
In order to suppress (or mitigate) vortex pairing, we therefore need to design a controller that rejects perturbations whose projection onto $\vxi(t)$ is non-zero.
Given the full-order state $\vx'(t) = \mPhi(t)\vq(t)$, and letting $\vxi(t)$ be normalized such that
\begin{equation}
\label{eq:vxi_normalized}
    \frac{1}{2 T}\int_0^{2 T}\vxi(t)^*\vxi(t)\,\mathrm{d}t =  \sum_{k\in\mathbb{Z}}\hat\vxi_k^*\hat\vxi_k = 1,
\end{equation}
the projection of $\vx'(t)$ onto $\vxi(t)$ is given by $\vxi(t)^*\mPhi(t)\vq(t)$.
This information may be encoded into the LQR problem by choosing $\mQ_q(t)$ as follows,
\begin{equation}
\label{eq:mQ_lqr}
    \mQ_q(t) = \alpha \left(\mPhi(t)^*\vxi(t)\vxi(t)^*\mPhi(t)\right),
\end{equation}
where $\alpha$ is a positive scalar.
The matrix $\mQ_q(t)$ now contributes to the cost function $\mathcal{J}_{LQR}$ by measuring the projection of the full-state onto the ``most dangerous" mode.
Consequently, the resulting optimal feedback law will try to change the dynamics by minimizing the projection of the state onto $\vxi(t)$.
The same rationale was applied in the $Re = 1500$ case.

In order to implement the feedback law discussed above, it is necessary to design a state estimator (or observer), which, given the available sensor measurements, computes a state estimate $\widetilde\vq(t)$.
The desired control input will then be given by $u(t) = -\mK(t)\widetilde\vq(t)$.
Here, we design an observer using the linear quadratic estimator (LQE) approach, which assumes that the reduced state $\vq(t)$ and the measured output $\vy(t)$ are corrupted by Gaussian noise.
More specifically, we suppose that $\vq(t)$ and $\vy(t)$ are governed by
\begin{equation}
    \begin{aligned}
    \ddt\vq(t) &= \mA_r(t)\vq(t) + \mB_r(t)u(t) + \overline{\mB}_r(t)d(t) \\
    \vy(t) &= \mC_r(t)\vq(t) + \vn(t),
    \end{aligned}
\end{equation}
where the disturbance $d(t)$ and the sensor noise $\vn(t)$ are zero-mean
Gaussian processes with covariance $\mathbb{E}[d(t)d(\tau)]=\beta\,
\delta(t-\tau)$ and $\mathbb{E}[\vn(t)\vn(\tau)^*]=\mQ_n \delta(t-\tau)$, respectively.
The linear operator $\overline{\mB}_r(t)$ is $m T$-periodic and it is chosen by the user to model how the disturbance $d(t)$ enters the dynamics.
We will elaborate on this choice shortly.
It may be shown that the optimal state estimate $\widetilde\vq(t)$ is governed by the dynamics below
\begin{equation}
\label{eq:obsv}
    \ddt \widetilde\vq(t) = \left(\mA_r(t)-\mL(t)\mC_r(t)\right)\widetilde\vq(t) + \mB_r(t)u(t) + \mL(t)\vy(t),
\end{equation}
where the $m T$-periodic matrix $\mL(t)$ is chosen to minimize the expected
estimation error
\begin{equation*}
  \lim_{t\to\infty} \mathbb{E}\big[\big\|\vq(t) - \widetilde\vq(t)\big\|\big].
\end{equation*}
In our implementation, we choose $\mQ_n$ to be a diagonal matrix with entries
\begin{equation}
    [\mQ_n]_{i,i} = \frac{\max_t |\mC_r(t)|}{\max_t |\mC_{r,i}(t)|},
\end{equation}
where $\mC_{r,i}$ denotes the $i$th row of the output matrix $\mC_r(t)$.
This ensures that the resulting estimator responds equally strongly (or weakly) to changes in each measured output.
Finally, we design the matrix $\overline{\mB}_r(t)$ once again by leveraging the results from \cite{padovan2022}.
According to that analysis, the only external disturbances that have a
measurable effect on the flow are those that align with the aforementioned mode~$\vxi(t)$.  Therefore, we choose
\begin{equation}
    \overline{\mB}_r(t) = \mPsi(t)^*\vxi(t).
\end{equation}
That is, we model the disturbances that enter the dynamics via the ``most dangerous'' mode and we disregard all the other ones.
A block diagram of the observer-based feedback configuration is shown in figure \ref{fig:block_diags}.

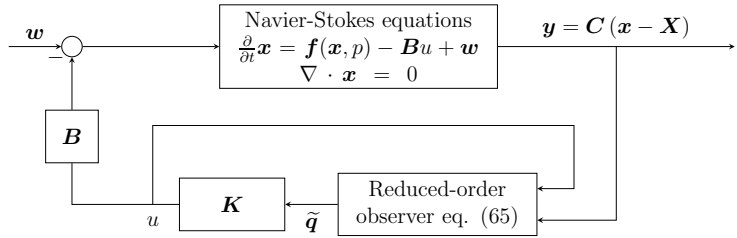
\begin{figure}
\centering
\scalebox{0.7}
{
\begin{tikzpicture}[auto, node distance=1cm,>=latex',>=stealth]
    \node [input, name=input] {};
    \node [sum, right=of input] (sum) {};
    \node [block, right=4 of input, text width=5cm,align=center] (plant)
    {Navier-Stokes equations $\frac{\partial}{\partial t}\vx = \vf(\vx,p) - \mB u + \vw$ $\nabla \cdot \vx = 0$};
    \node [output, right=4.5 of plant] (output) {};
    \node at ($(plant)+(1.5,-3)$) [block,text width=3.5cm,align=center] (obsv) {Reduced-order observer eq. \eqref{eq:obsv}};
    \node [block, left=1 of obsv] (feedback) {$\mK$};
    \draw [draw,->] (input) -- node {$\vw$} (sum);
    \draw [->] (sum) -- (plant);
    \draw [->] (plant) -- node [name=y] {$\vy =\mC \left(\vx - \vX\right)$}(output);
    \draw [->] (y) |- ([yshift=-0.3cm]obsv.east) ;
    \draw [draw,->] (obsv) -- node {$\widetilde\vq$} (feedback);
    \node [block_small, below=1 of sum] (B) {$\mB$}; 
    \draw [-] (feedback) -| (B);
    \draw [->] (B) -- node[pos=0.99] {$-$} (sum);
    \coordinate[left=0.5 of feedback,name=node0];
    \coordinate[above=1.5 of node0,name=node1];
    \coordinate[right=8 of node1,name=node2];
    \draw[->] (node0) -- (node1) -- (node2) |- ([yshift=0.3cm]obsv.east);
    \node [label={[xshift=-0.0cm, yshift=-0.7cm]$u$},left=0.35 of feedback]{};
\end{tikzpicture}
}
\caption{Block diagram for the observer-based feedback configuration. The plant
  (labelled ``Navier-Stokes equations'') represents equations \eqref{eq:ns-all}
  plus the additional feedback term $-\mB u$ and some external forcing input
  (or disturbance) $\vw$.}
\label{fig:block_diags}
\end{figure}

\subsection{Suppressing vortex pairing at $Re=1250$}
\label{subsec:results_Re1250}

We begin by considering the case $Re =1250$, where the $T$-periodic base flow is linearly stable.
However, as previously mentioned, almost every small perturbation triggers vortex pairing.
Here, we wish to design a reduced-order feedback controller and observer to suppress vortex pairing.
We begin by computing a reduced-order model using the balancing procedure described in algorithms \ref{alg:compute_factor} and \ref{alg:compute_balancing}.
The balanced model has period $2 T$ (recall that it is possible that the ROM has a higher period than the underlying full-order model) and we select model size $r = 6$.
The truncation rank is chosen based on the decay of the singular values and on the predictive accuracy of the ROM, both shown in \ref{app:roms_Re1250}.
Using the reduced-order model, we design a feedback controller using the strategy discussed in section \ref{subsec:controller_observer_design}, and we select $\mQ_q(t)$ in \eqref{eq:mQ_lqr} with
\begin{equation}
\label{eq:alpha_lqr}
    \alpha = \frac{ 10^{-3}}{\max_{t}\left(\mPhi(t)^*\vxi(t)\vxi(t)^*\mPhi(t)\right)}.
\end{equation}
We also design an estimator and we choose $\beta=10$.
At this point we are ready to verify if we can suppress (or at least mitigate) the vortex pairing phenomenon.
Since we know that vortex pairing is driven almost exclusively by the mode $\vxi(t)$ in Padovan and Rowley (2022), we induce vortex pairing by forcing the nonlinear Navier-Stokes equations \eqref{eq:ns-all} with the external forcing input
\begin{equation}
\label{eq:forcing_vw}
    \vw(t) = 10^{-3}\vxi(t),\quad \vw(t) = \vw(t+2 T),
\end{equation}
where $\vxi(t)$ is normalized as in \eqref{eq:vxi_normalized}.
The initial condition to the nonlinear full-order simulation is taken to be a state of heavy vortex pairing, while the initial condition for the reduced-order observer is set to zero.
This is to emulate the realistic scenario where we do not have a good guess for the initial reduced-order state.
We then integrate the (nonlinear) Navier-Stokes equations with and without
feedback control, using the observer-based feedback configuration shown in figure~\ref{fig:block_diags}.
The results are shown in figures \ref{fig:sensor_output_Re1250} and~\ref{fig:snapshots_Re1250}.

\begin{figure}
\centering
\begin{minipage}{0.48\textwidth}
\begin{tikzonimage}[trim= 10 50 0 75,clip,width=0.92\textwidth]{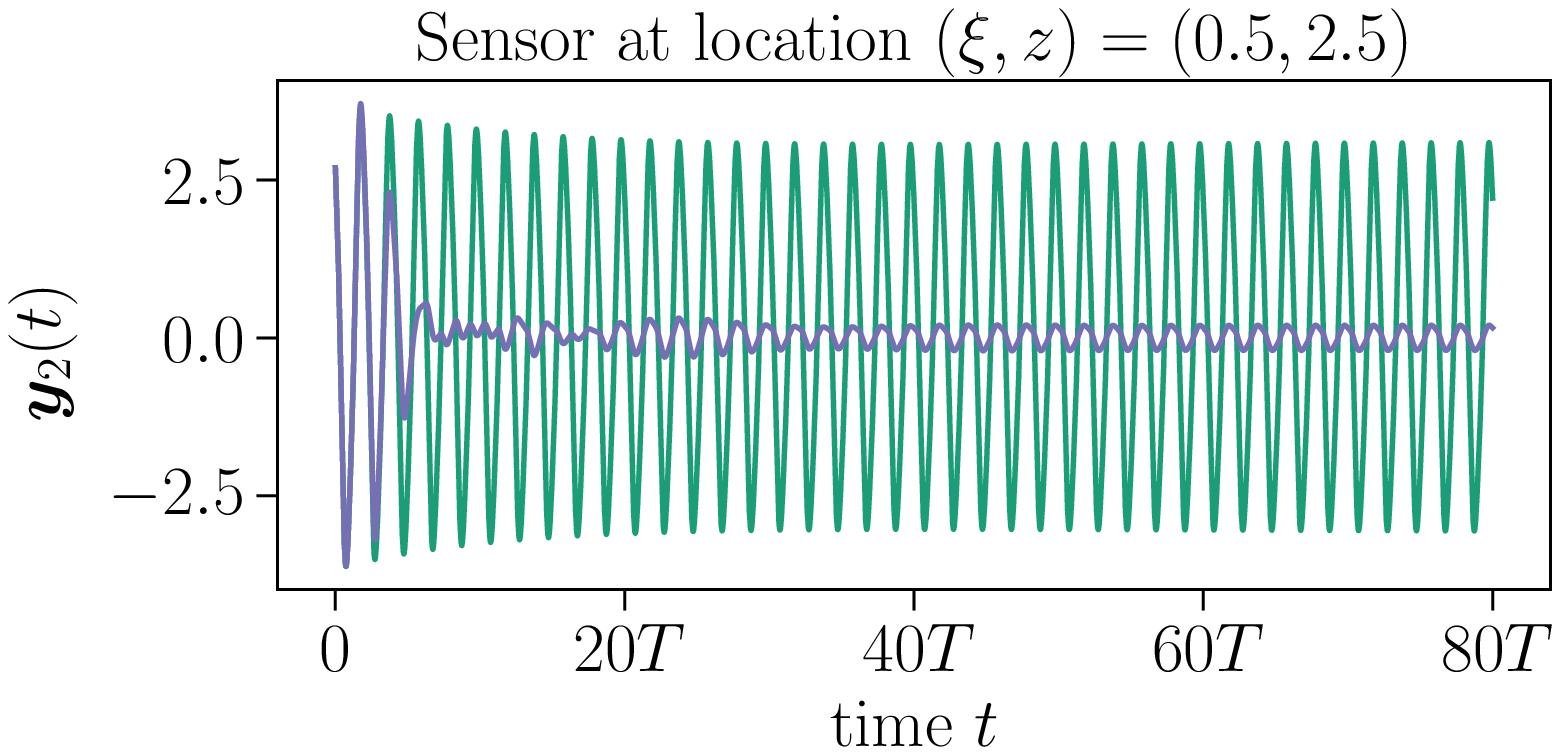}
\node at (0.90,0.78) {\small $\textit{(a)}$};
\end{tikzonimage}
\end{minipage}
\begin{minipage}{0.48\textwidth}
\begin{tikzonimage}[trim= 10 50 0 75,clip,width=0.92\textwidth]{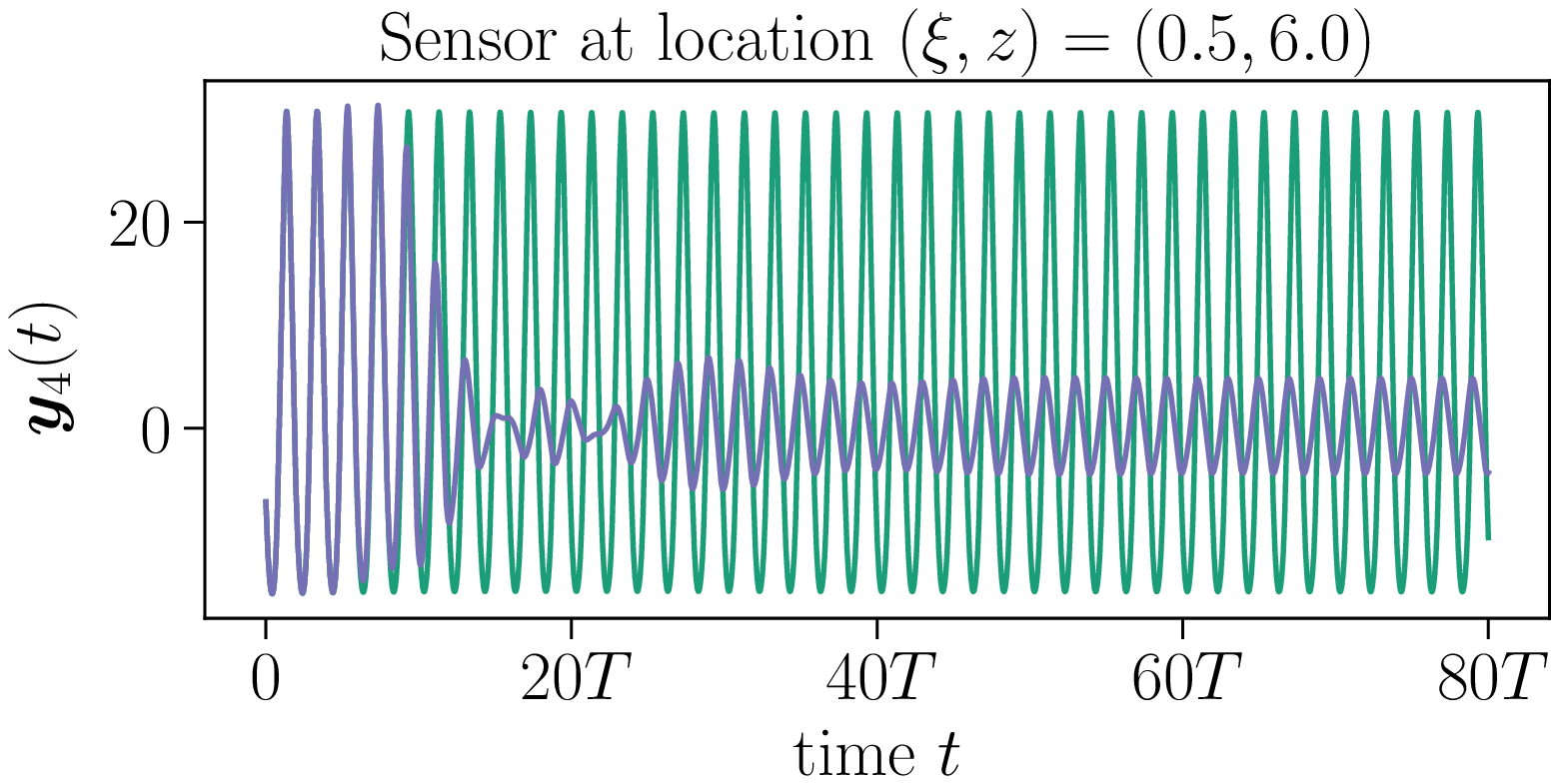}
\node at (0.90,0.78) {\small $\textit{(b)}$};
\end{tikzonimage}
\end{minipage}
\begin{minipage}{0.48\textwidth}
\begin{tikzonimage}[trim= 10 50 0 75,clip,width=0.92\textwidth]{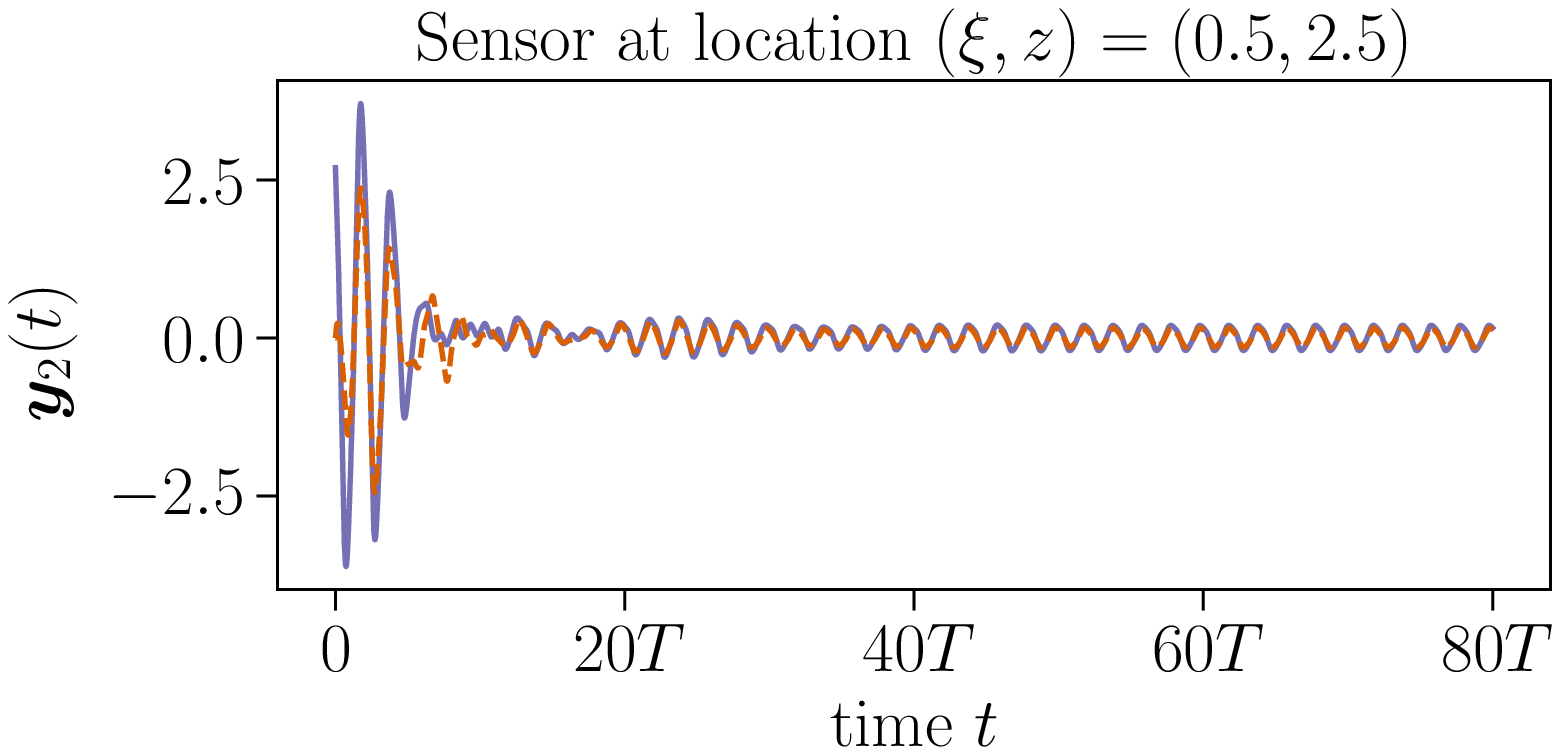}
\node at (0.90,0.78) {\small $\textit{(c)}$};
\end{tikzonimage}
\end{minipage}
\begin{minipage}{0.48\textwidth}
\begin{tikzonimage}[trim= 10 50 0 75,clip,width=0.92\textwidth]{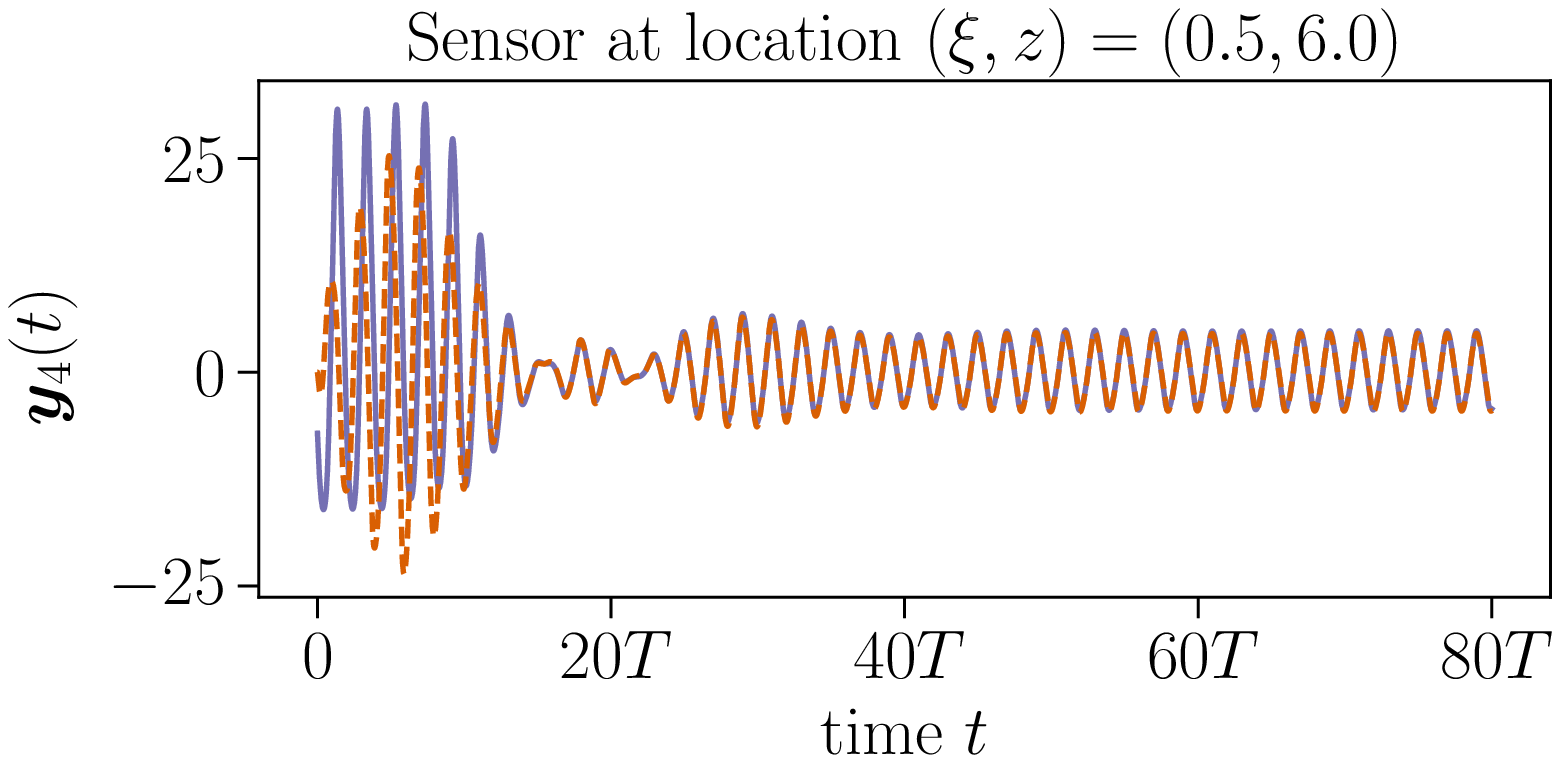}
\node at (0.90,0.78) {\small $\textit{(d)}$};
\end{tikzonimage}
\end{minipage}
\caption{$Re = 1250$. In panels \textit{(a)} and \textit{(b)} we show the sensor
  measurement from two of the four sensors from the simulation without control
  (green) and with feedback control (blue).
In panels \textit{(c)} and \textit{(d)} we show the sensor measurement from the
simulation with control (blue) and the predicted sensor measurement from the reduced-order estimator (orange).
}
\label{fig:sensor_output_Re1250}
\end{figure}

\begin{figure}
\centering
\begin{tikzonimage}[trim= 5 50 0 80,clip,width=0.5\textwidth]{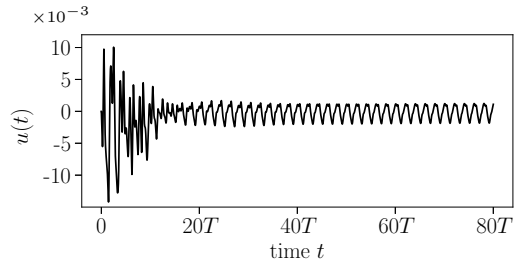}
\node at (0.1,1) {\tiny $\times 10^{-3}$};
\end{tikzonimage}
\caption{Control input $u(t)$ for the closed-loop simulation at $Re = 1250$.}
\label{fig:input_Re1250}
\end{figure}

We begin with a brief discussion of figure \ref{fig:sensor_output_Re1250}, where we show the time history of the measurements from sensors 2 and 4 located at $(\xi,z)=(0.5,2.5)$ and $(\xi,z)=(0.5,6)$, respectively.
In panels \textit{(a)} and \textit{(b)} we show the measured velocity of the flow with no control (green) and with observed-based feedback (blue).
As hoped, we observe a significant reduction in the oscillation amplitude at both sensor locations when the feedback is active, meaning that the controller is successfully rejecting the disturbances that we are injecting into the flow.
Panels \textit{(c)} and \textit{(d)} show the measured velocity of the flow with observer-based feedback (blue) and the predicted velocity from the reduced-order observer.
Here, we see that at early times there are some discrepancies due to the fact that the nonlinear simulation is initialized with a non-zero initial condition, while the observer is initialized at zero.
At later times, we see that the predicted output converges to the ground truth output, and this explains the success of the reduced-order controller/observer system in suppressing the oscillatory behavior of the flow.

In figure \ref{fig:snapshots_Re1250} we show representative vorticity snapshots
from the simulations with and without feedback control.
Without control, the vortex rings pair at an axial location around $z\approx 5$.
When the flow is controlled, we see that vortex pairing is significantly
delayed/mitigated, and vortex rings pair further downstream, at around $z\approx 8$.
We conclude by observing that this controller/observer pair was able to modify the flow dynamics with a control input that never exceeds $O(10^{-2})$ and actually remains below $5\times 10^{-3}$ for most of the times (see figure \ref{fig:input_Re1250}).
This is indicative of the fact that the actuator is placed at a location where we have large control authority, and we therefore only require small-amplitude perturbations to modify the flow behavior.
In a more practical setting, this means that this controller requires a low external energy supply, since it should not be energetically expensive to provide a velocity perturbation with magnitude equal to one thousandth of the flow characteristic velocity.

\begin{figure}
\centering
\begin{minipage}{0.48\textwidth}
\begin{tikzonimage}[trim= 5 10 5 20,clip,width=\textwidth]{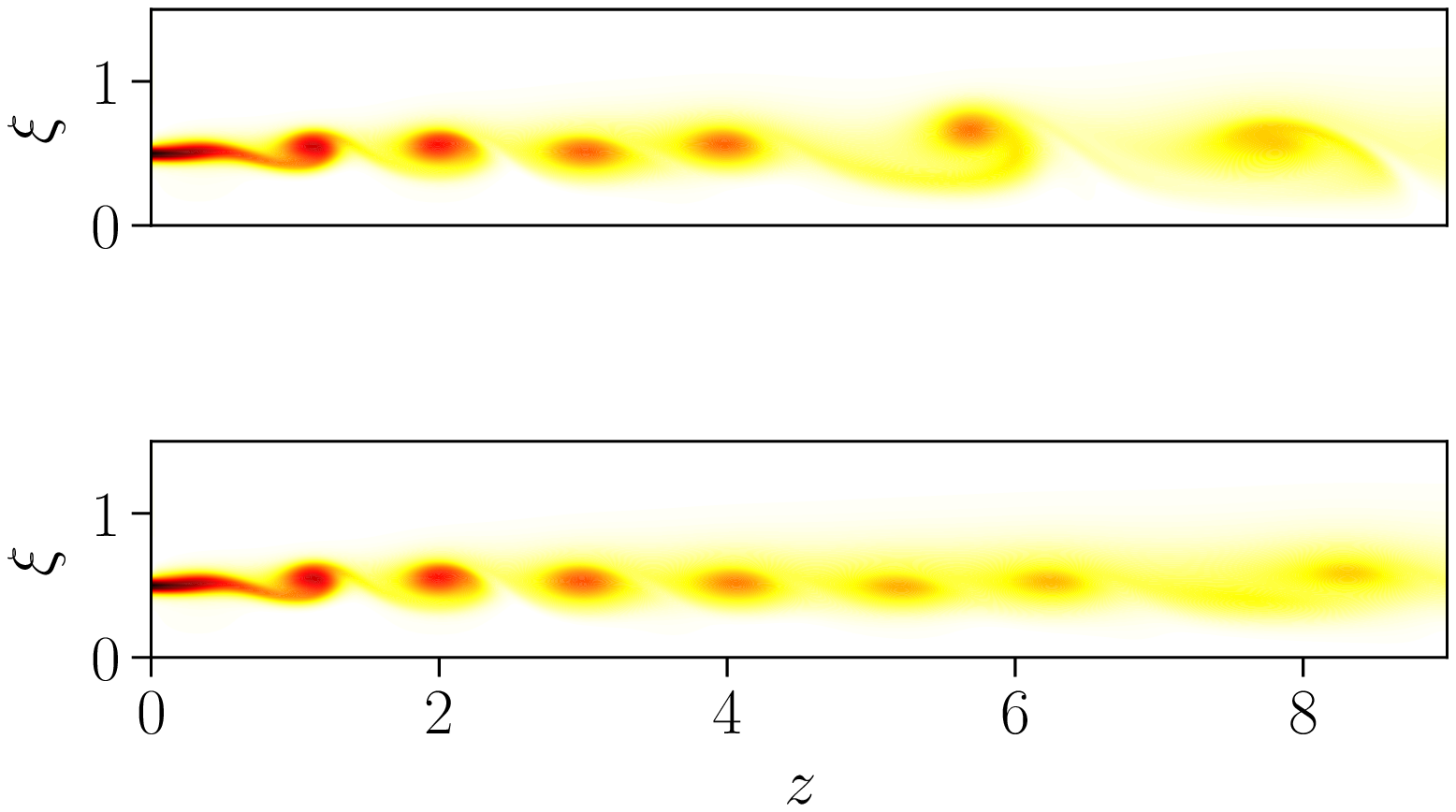}
\node at (0.55,1.05) {\small Vorticity at time $t \approx 79 T$};
\node at (0.55,0.92) {\small no control};
\node at (0.55,0.5) {\small with control};
\end{tikzonimage}
\end{minipage}
\begin{minipage}{0.48\textwidth}
\begin{tikzonimage}[trim= 5 10 5 20,clip,width=\textwidth]{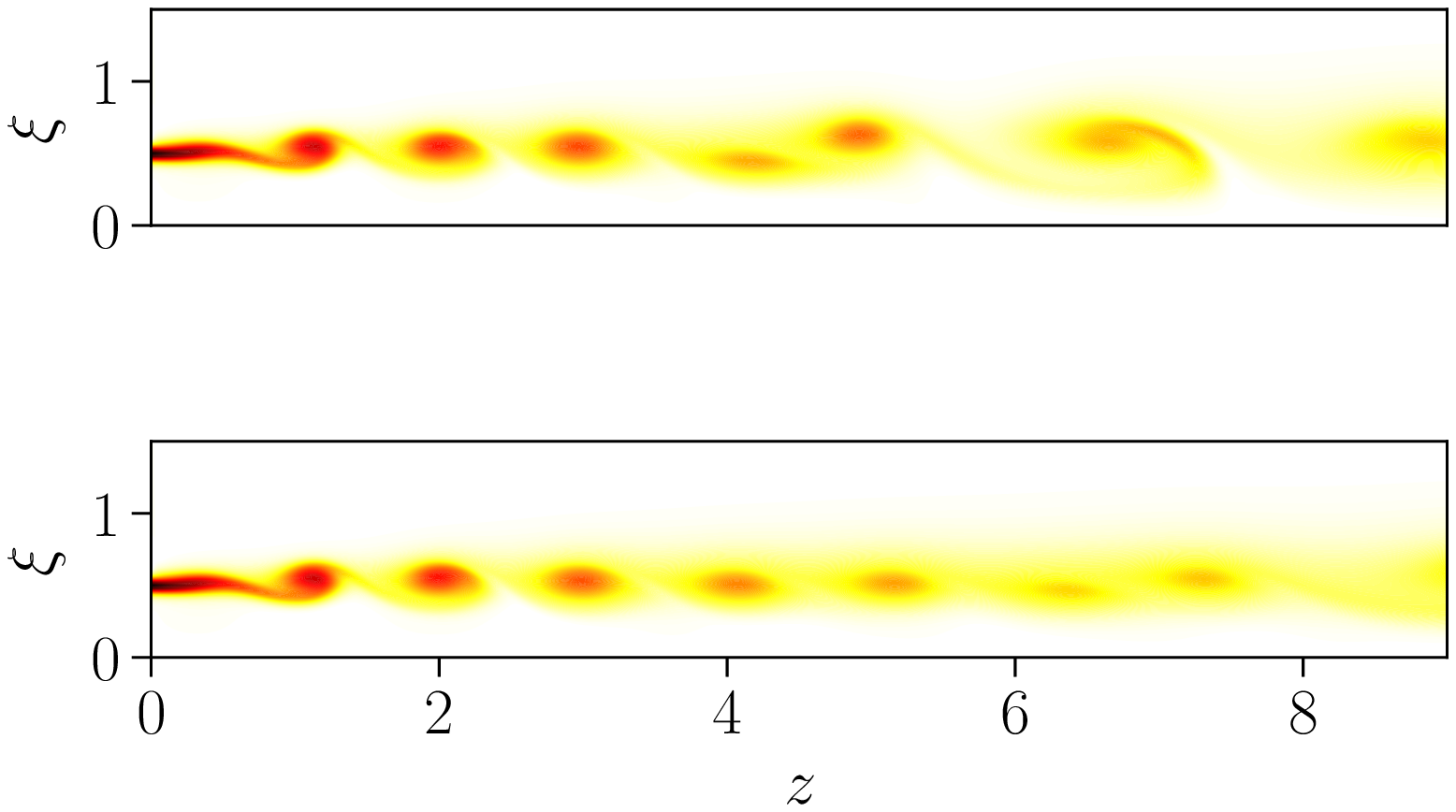}
\node at (0.55,1.05) {\small Vorticity at time $t \approx 80 T$};
\node at (0.55,0.92) {\small no control};
\node at (0.55,0.5) {\small with control};
\end{tikzonimage}
\end{minipage}
\caption{$Re = 1250$. Vorticity snapshots at times $t \approx 79 T$ and $t \approx 80 T$ from the
  response with and without feedback control. The colorbar is the same as in
  figure~\ref{fig:bflow}.}
\label{fig:snapshots_Re1250}
\end{figure}

\subsection{Suppressing vortex pairing at $Re = 1500$}

We now consider $Re = 1500$, for which the $T$-periodic base flow is unstable,
and the flow spirals onto a $2T$-periodic limit cycle characterized by pairing
vortex rings.
As before, we compute a ROM of dimension $r=6$ (see \ref{app:roms_Re1500} for details), and we design a controller as well as an observer.
We choose the LQR weight $\mQ_{q}(t)$ as in \eqref{eq:mQ_lqr} with $\alpha$ given in \eqref{eq:alpha_lqr}.
Similarly, the LQE weight is chosen as $\beta = 10$.
As in the previous section, we externally force the flow with the ``most dangerous" forcing profile \eqref{eq:forcing_vw}, except that the magnitude is set to $10^{-4}$.
Since the base flow is unstable, vortex pairing will naturally occur even without the external forcing input.
Here, however, we use this input to emulate the presence of external disturbances that perturb the flow on top of the underlying instability.
The initial condition for the nonlinear full-order simulation is taken to be a state of heavy vortex pairing, while the initial condition for the reduced-order observer is set to zero.
As explained in the previous section, this is a realistic choice based on the fact that we often lack knowledge of the initial state of the system.
We then integrate the nonlinear Navier-Stokes equations in open loop and in closed-loop with the observed-based feedback configuration shown in figure \ref{fig:block_diags}.
Results analogous to those shown in the previous section are shown in figure \ref{fig:sensor_output_Re1500}.

In this figure we see that the controller/observer pair is capable of suppressing the highly oscillatory behavior of the flow at all four sensor locations.
Moreover, we see that the prediction of the measured output provided by the reduced-order estimator agrees well with the ground-truth measurements, especially at long times (see figures \ref{fig:sensor_output_Re1500}c and \ref{fig:sensor_output_Re1500}d).
As in the previous section, the initial discrepancies arise because the observer was initialized with a zero initial condition, while the full-state had a non-zero initial condition.

\begin{figure}
\centering
\begin{minipage}{0.48\textwidth}
\begin{tikzonimage}[trim= 10 50 0 75,clip,width=0.92\textwidth]{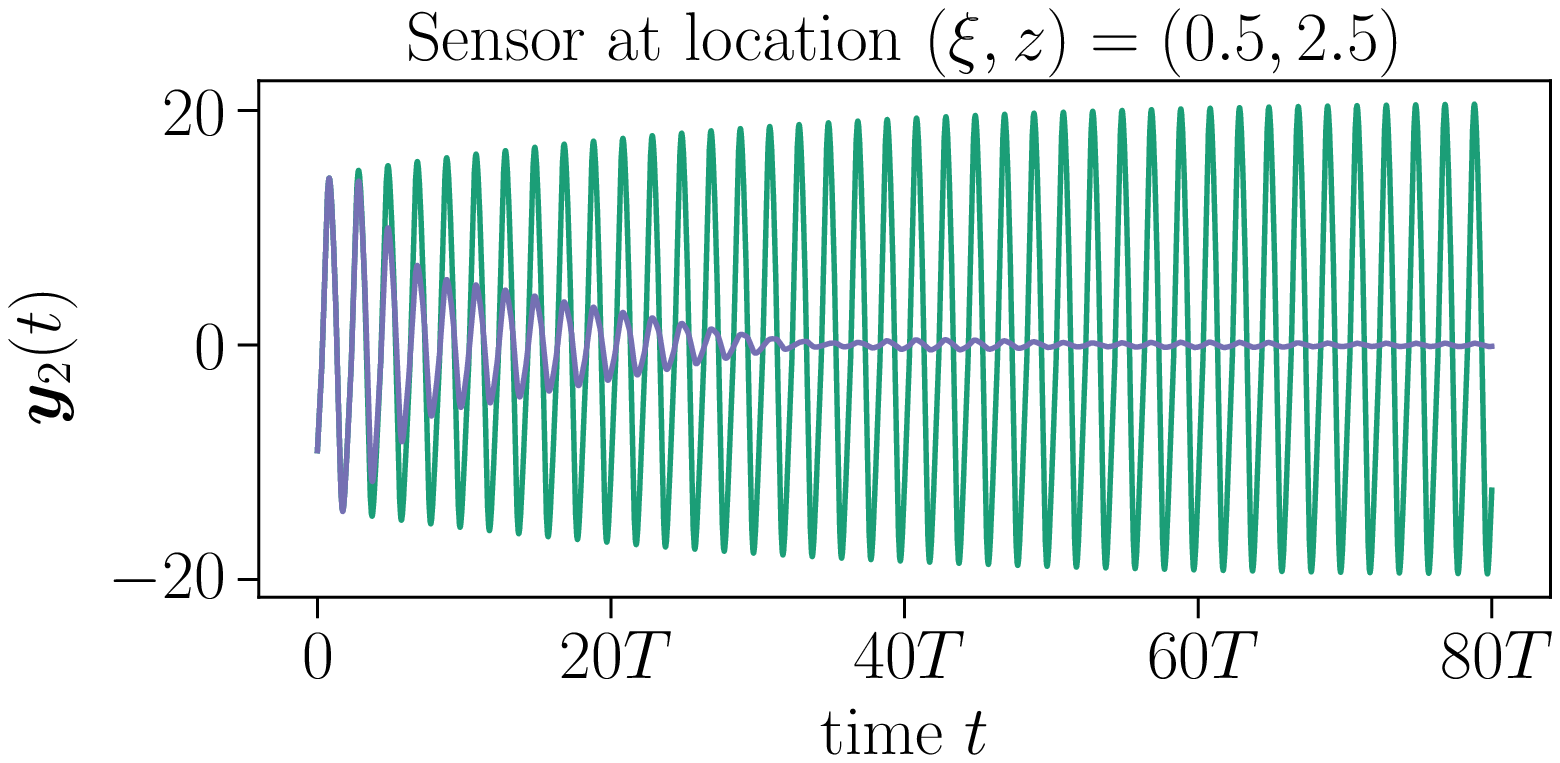}
\node at (0.90,0.78) {\small $\textit{(a)}$};
\end{tikzonimage}
\end{minipage}
\begin{minipage}{0.48\textwidth}
\begin{tikzonimage}[trim= 10 50 0 75,clip,width=0.92\textwidth]{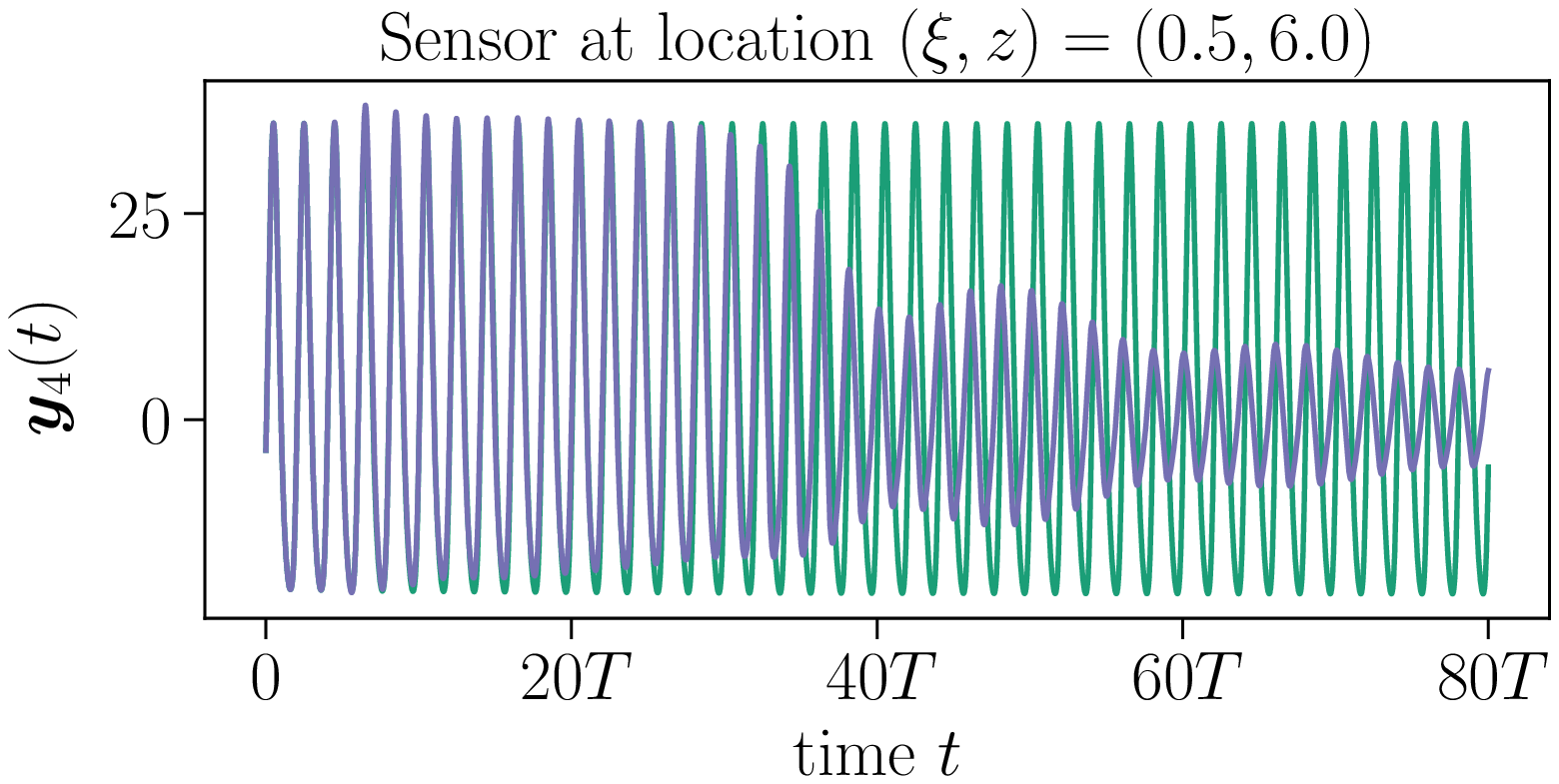}
\node at (0.90,0.78) {\small $\textit{(b)}$};
\end{tikzonimage}
\end{minipage}
\begin{minipage}{0.48\textwidth}
\begin{tikzonimage}[trim= 10 50 0 75,clip,width=0.92\textwidth]{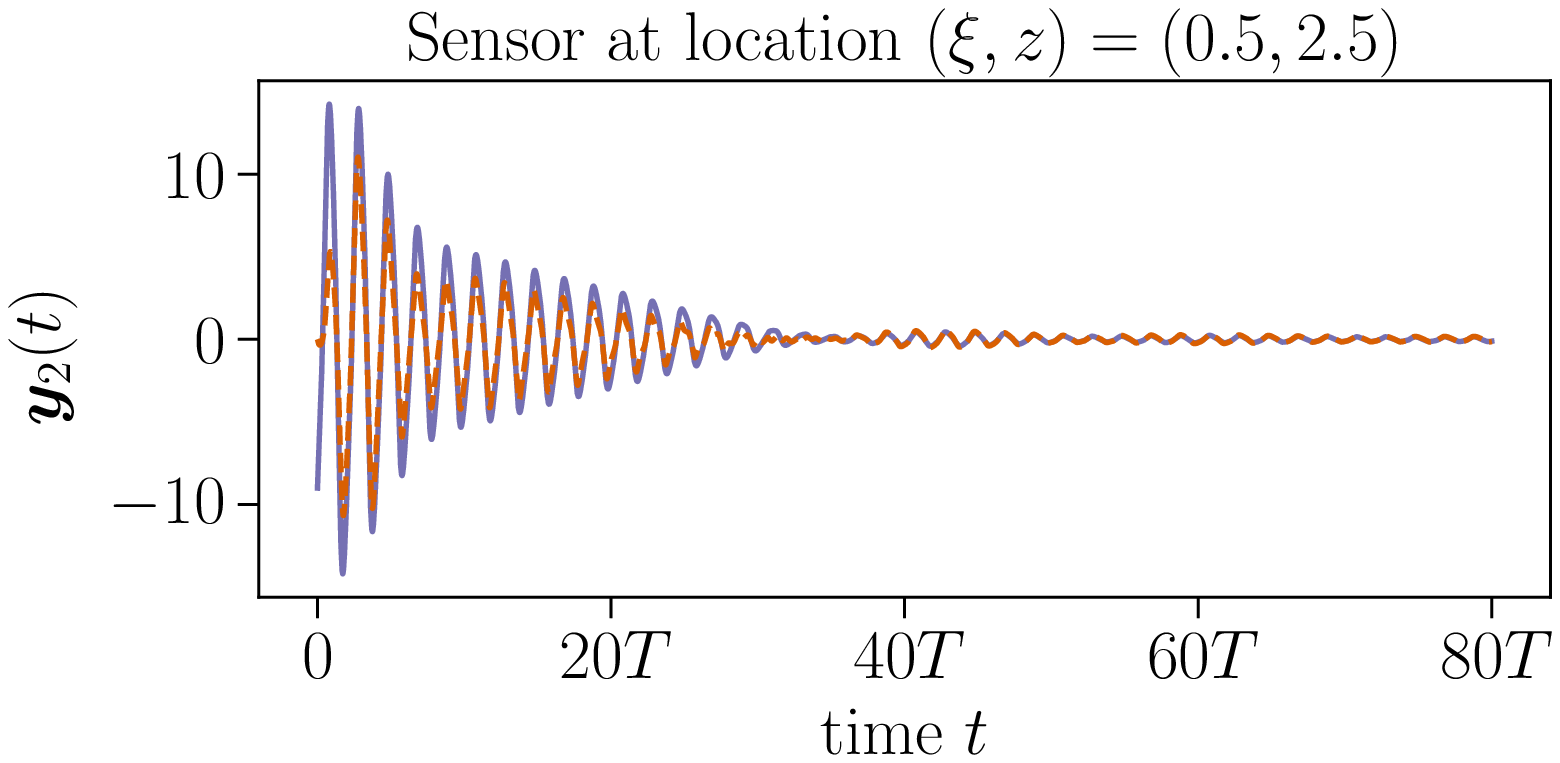}
\node at (0.90,0.78) {\small $\textit{(c)}$};
\end{tikzonimage}
\end{minipage}
\begin{minipage}{0.48\textwidth}
\begin{tikzonimage}[trim= 10 50 0 75,clip,width=0.92\textwidth]{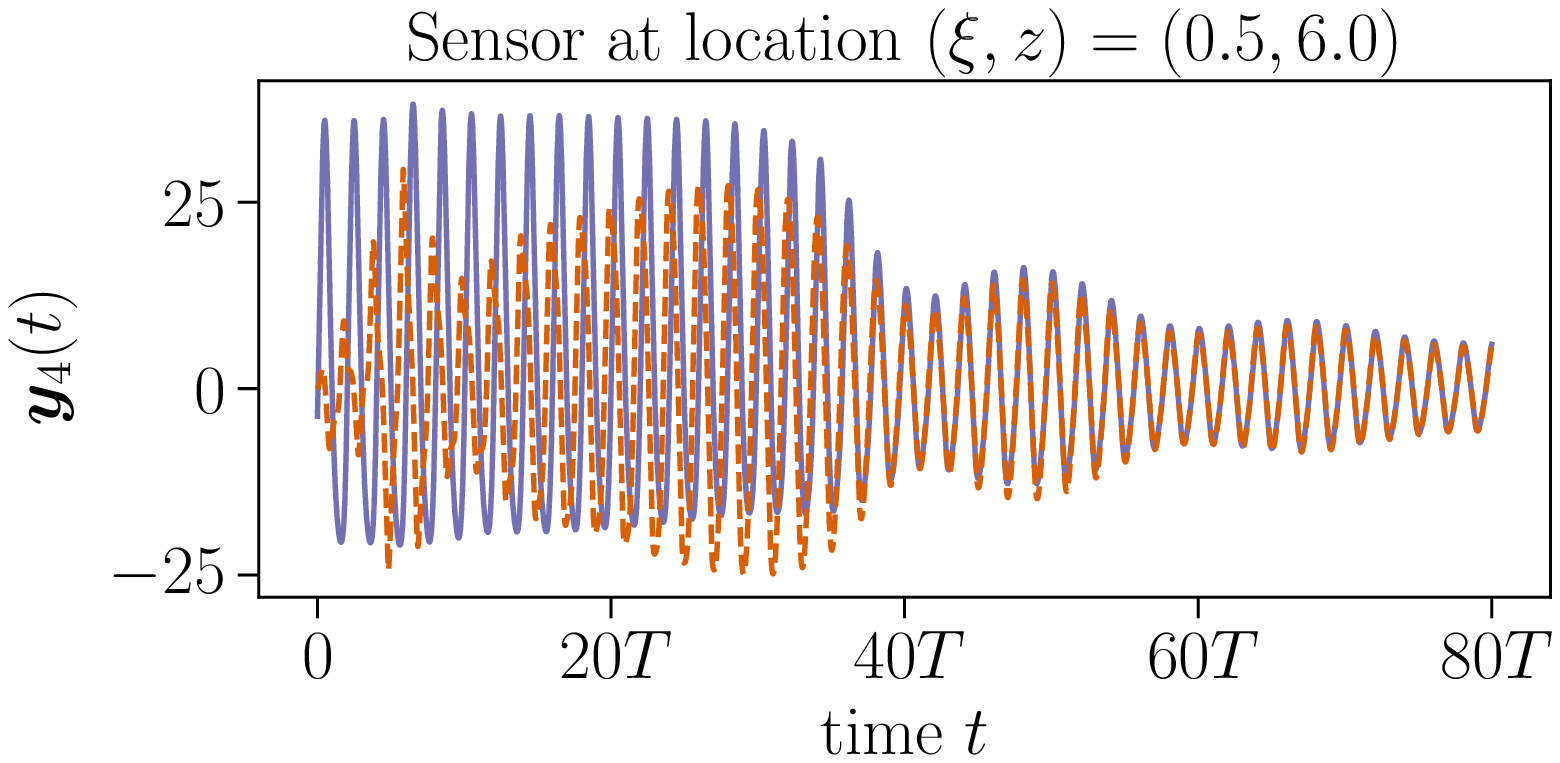}
\node at (0.90,0.78) {\small $\textit{(d)}$};
\end{tikzonimage}
\end{minipage}
\caption{Analog of figure \ref{fig:sensor_output_Re1250} for $Re = 1500$.
}
\label{fig:sensor_output_Re1500}
\end{figure}

\begin{figure}
\centering
\begin{tikzonimage}[trim= 5 50 0 80,clip,width=0.5\textwidth]{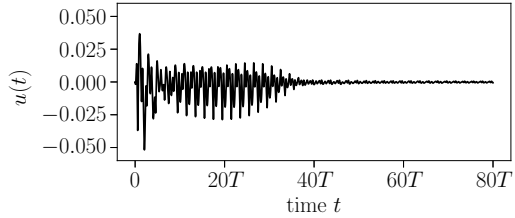}
\end{tikzonimage}
\caption{Control input $u(t)$ for the closed-loop simulation at $Re = 1500$.}
\label{fig:input_Re1500}
\end{figure}

Representative snapshots of the simulations with and without feeback control are shown in figure~\ref{fig:snapshots_Re1500} at two different time instances.
Remarkably, we see that while the uncontrolled flow exhibits strong vortex
pairing (much stronger than the $Re = 1250$), the flow with feedback control does not.
This means that the controller/observer pair successfully managed to suppress the underlying instability and to reject the most dangerous disturbance that we are injecting into the flow.
Finally, we see from figure~\ref{fig:input_Re1500} that the feedback law generates a control input with magnitude $O(10^{-2})$ at early times, but once the upstream oscillations (see, e.g., figure~\ref{fig:sensor_output_Re1500}a) have been suppressed, the required control input drops by an order of magnitude.
We would like to point out that even at early times, the demanded control input is significantly lower in magnitude than the flow characteristic velocity.

\begin{figure}
\centering
\begin{minipage}{0.48\textwidth}
\begin{tikzonimage}[trim= 5 10 5 20,clip,width=\textwidth]{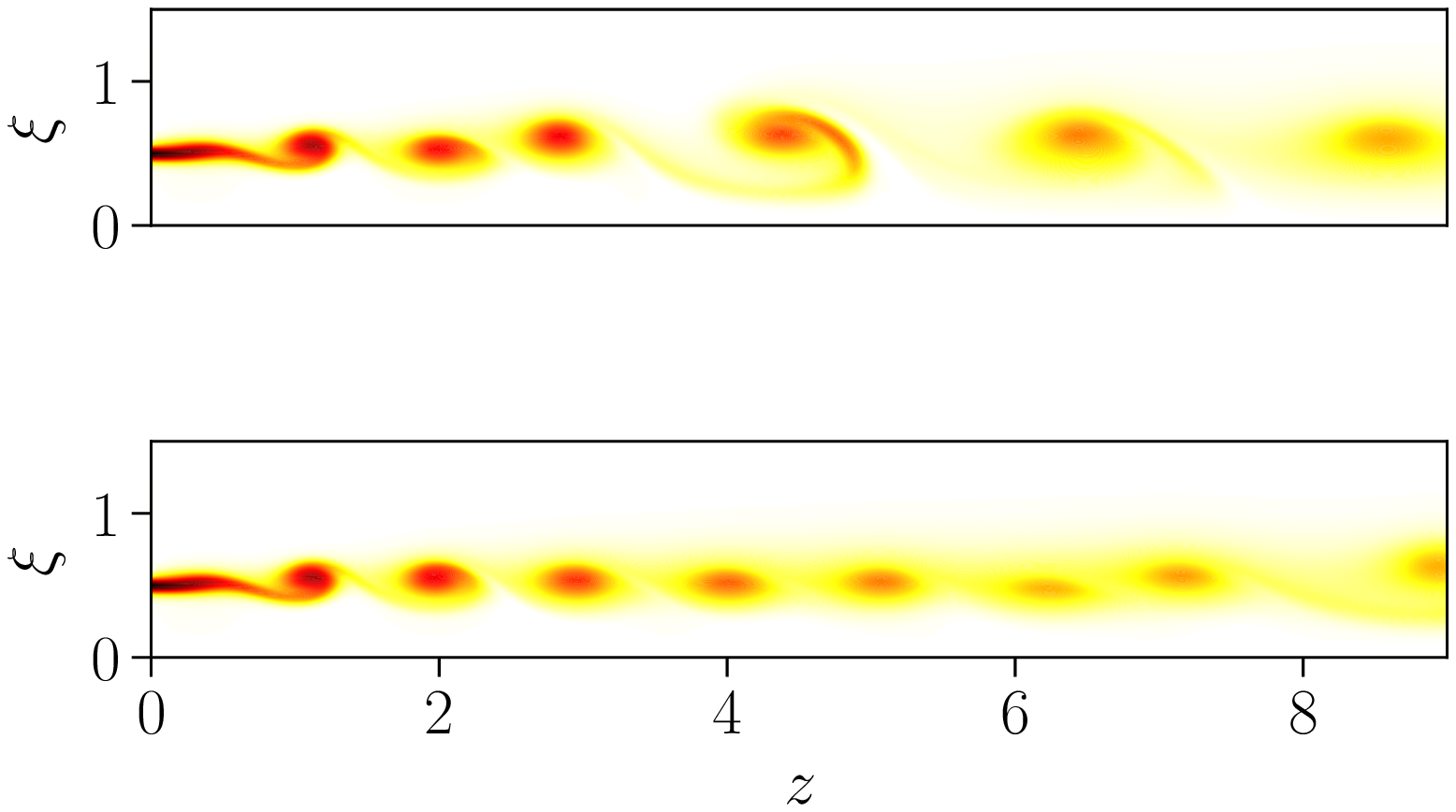}
\node at (0.55,1.05) {\small Vorticity at time $t \approx 79 T$};
\node at (0.55,0.92) {\small no control};
\node at (0.55,0.5) {\small with control};
\end{tikzonimage}
\end{minipage}
\begin{minipage}{0.48\textwidth}
\begin{tikzonimage}[trim= 5 10 5 20,clip,width=\textwidth]{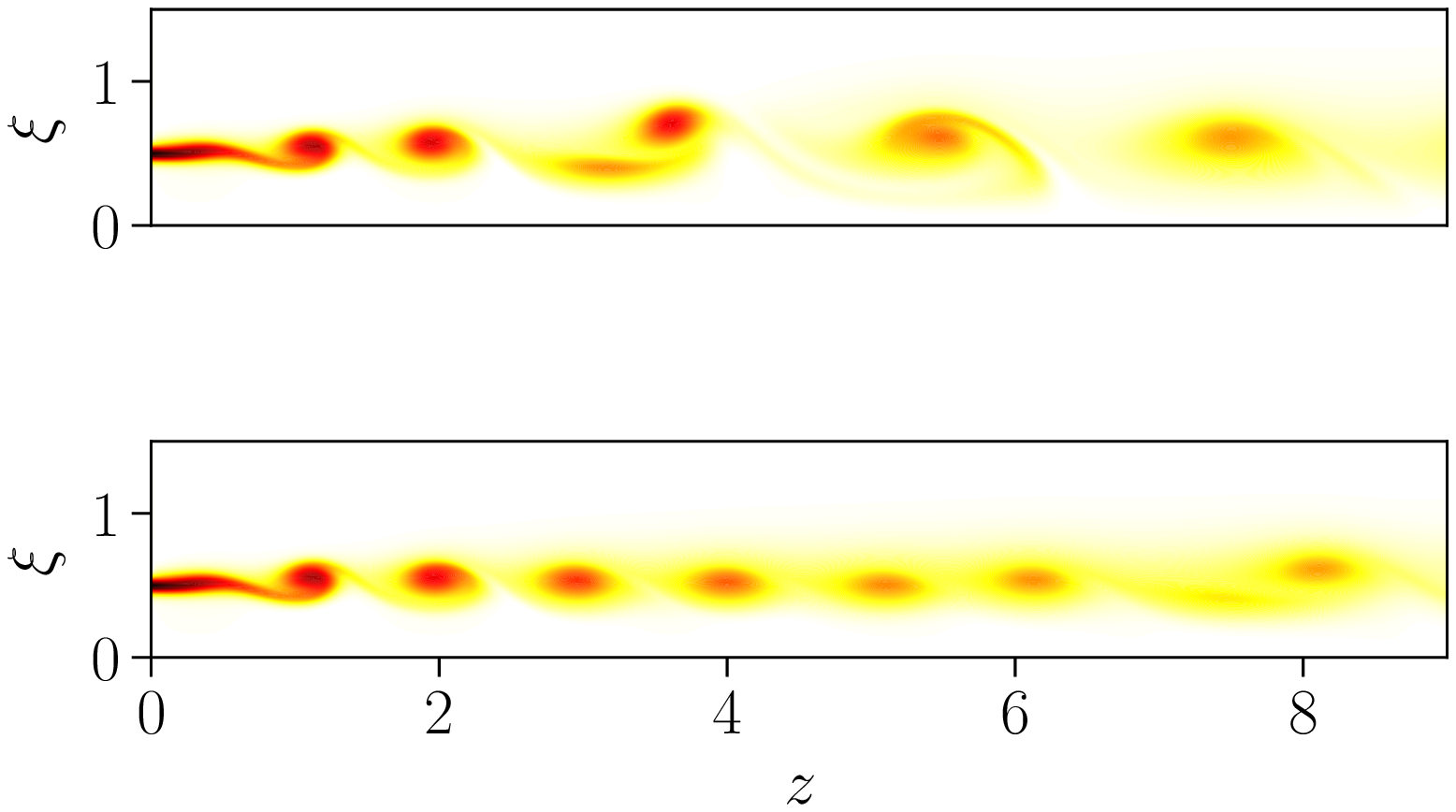}
\node at (0.55,1.05) {\small Vorticity at time $t \approx 80 T$};
\node at (0.55,0.92) {\small no control};
\node at (0.55,0.5) {\small with control};
\end{tikzonimage}
\end{minipage}
\caption{Analog of figure \ref{fig:snapshots_Re1250} at $Re = 1500$.}
\label{fig:snapshots_Re1500}
\end{figure}

\section{Conclusion}

In this paper we perform continuous-time balanced truncation for time-periodic systems using the frequency-domain representation of the reachability and observability Gramians.
We have seen that these frequential Gramians are well-defined both when the underlying system is stable and when it is unstable.
Moreover, when have seen that computing the Gramians using their frequency-domain representation can offer computational benefits, especially if the dynamics exhibit slowly-decaying transients.
We demonstrated this approach on a periodically-forced axisymmetric jet at
Reynolds numbers $Re = 1250$ and $Re = 1500$ (corresponding to stable and
unstable equilibria), and in both cases we used the balanced model to design reduced-order controllers and observers to suppress the vortex pairing phenomenon.

\section*{Acknowledgements}
This work was supported by the Air Force Office of Scientific Research, award
FA9550-19-1-0005.

\appendix

\section{Proofs}
\label{app:proofs}

\subsection{Proof of Proposition \ref{prop:grams_as_outer_product}}
\label{app_subsec:proof_gram_as_outer_products}

We prove the proposition for $\mP(t)$, since the result for $\mQ(t)$ follows analogously.
Let us define the quantity
\begin{equation}
\label{eq:hat_mP_k}
    \hat\mP_k = \sum_{\substack{m, l\in\mathbb{Z}\\m=k+l}}\hat\mP_{m-l} \coloneqq \sum_{\substack{m, l\in\mathbb{Z}\\m=k+l}}\frac{1}{2\pi}\int_{-\infty}^{\infty}\mZ_m(\gamma)\mZ_l(\gamma)^{*}\mathrm{d}\gamma.
\end{equation}
Given the definition of $\mZ_m(\gamma)$ in \eqref{eq:mZ} and using Proposition \ref{prop:unstable_lyap_eqtn}, it follows that $\hat\mP_{m-l}$ satisfies the algebraic Sylvester equation below
\begin{equation}
    \left(-i m \omega\mI + \mJ\right) \hat\mP_{m-l} +  \hat\mP_{m-l}\left(i l \omega\mI+ \mJ^*\right) + \mathcal{P}_s{\widetilde\mB}_{m} {\widetilde\mB}_l^*\mathcal{P}_s - \mathcal{P}_u{\widetilde\mB}_{m} {\widetilde\mB}_l^*\mathcal{P}_u = 0.
\end{equation}
Summing over $m$ and $l$ as in \eqref{eq:hat_mP_k}, writing $m = k+l$, and using the linearity of the Sylvester equation, one can see that $\hat\mP_k$ satisfies \eqref{eq:sylv_fwd_unstable}.
In other words, $\hat\mP_k$ is the $k$th Fourier coefficient of $\mP(t)$.
Writing $\mP(t)$ as
\begin{equation}
    \mP(t) = \sum_{k\in\mathbb{Z}}\hat\mP_k e^{i k \omega t} = \sum_{k\in\mathbb{Z}}\sum_{\substack{m,l\in\mathbb{Z}\\m=k+l}}\hat\mP_{m-l} e^{i (m-l) \omega t} = \sum_{m,l\in\mathbb{Z}}\hat\mP_{m-l}e^{i(m-l)\omega t}
\end{equation}
concludes the proof.

\subsection{Proof of Proposition \ref{prop:gram_factors_hr}}
\label{app_subsec:proof_gram_factors_hr}

We need to show that
\begin{equation}
\label{eq:mH_equiv}
    \mH_{k,j}(\gamma) = \sum_{m\in\mathbb{Z}}\mV_{k-m}\left(i\gamma\mI -(-i m\omega\mI +\mJ)\right)^{-1}\mW^*_{j-m}.
\end{equation}
Let us start from \eqref{eq:sys_ltp_lti} and write $\vz(t)$ and $\vu(t)$ as EMP signals (as in \eqref{eq:emp}), to obtain
\begin{equation}
    \vz_{m+\gamma} = \left(i\gamma\mI -(-i m\omega\mI +\mJ)\right)^{-1} \sum_{j,l\in\mathbb{Z}}\mW^*_{j-m}\mB_{j-l}\vu_{l+\gamma}.
\end{equation}
Using $\vx(t) = \mV(t)\vz(t)$, the coefficient $\vx_{k+\gamma}$ is given by
\begin{equation}
\label{eq:vxz}
    \sum_{m\in\mathbb{Z}}\mV_{k-m}\vz_{m+\gamma} = \sum_{m,j,l\in\mathbb{Z}}\mV_{k-m}\left(i\gamma\mI -(-i m\omega\mI +\mJ)\right)^{-1} \mW^*_{j-m}\mB_{j-l}\vu_{l+\gamma}.
\end{equation}
Comparing \eqref{eq:vx} and \eqref{eq:vxz} shows that \eqref{eq:mH_equiv} indeed holds, and this concludes the proof.

\subsection{Proof of Proposition \ref{prop:alpha_gamma}}
\label{app_subsec:proof_alpha_gamma}

The existence of such integer $m$ is immediate.
Let us consider the quantity $\mZ_{\vx,k+\alpha}e^{i k\omega t} \coloneqq \sum_{j\in\mathbb{Z}}\mH_{k,j}(\alpha)\mB_j e^{i k\omega t}$, which, from the definition of $\mH$, satisfies
\begin{equation}
\label{eq:resp_alpha}
    i(\alpha + k\omega)\mZ_{\vx,k+\alpha}e^{i k\omega t} = \sum_{l\in\mathbb{Z}}\mA_{k-l}\mZ_{\vx,l+\alpha}e^{i k\omega t}  + \mB_k e^{i k\omega t}.
\end{equation}
Substituting $\alpha = \gamma + m\omega$ and manipulating the indices inside the sum, we obtain
\begin{equation}
\label{eq:resp_gam}
    i(\gamma + (k+m)\omega)\mZ_{\vx,(k+m)+\gamma} e^{i k\omega t} = \sum_{l\in\mathbb{Z}}\mA_{(k+m)-l}\mZ^{(j)}_{\vx,l+\gamma}e^{i k\omega t}+ \mB_k e^{i k\omega t}.
\end{equation}
Changing variables according to $n = k +m$, we have
\begin{equation}
    i(\gamma + n\omega)\mZ_{\vx,n+\gamma} e^{i (n-m)\omega t}= \sum_{l\in\mathbb{Z}}\mA_{n-l}\mZ_{\vx,l+\gamma}e^{i (n-m)\omega t} + \mB_{n-m}e^{i (n-m)\omega t}.
\end{equation}
Thus $\mZ_{\vx,k+\alpha} e^{i k \omega t} = \sum_{j\in\mathbb{Z}} \mH_{k,j}(\gamma)\mB_{j-m}e^{i(k-m)\omega t}$, and this concludes the proof.

\section{Reduced-order models at $Re = 1250$}
\label{app:roms_Re1250}

In this section we study the performance of different reduced-order models (ROMs) as a function of the model size.
Throughout, the Reynolds number $Re = 1250$, which gives us a linearly-stable periodic base flow with period $T$.

We begin by computing the Hankel singular values $\sigma_i(t)$ of the product $\mY(t)^*\mX(t)$ (see Algorithm \ref{alg:compute_balancing}).
From these singular values, we can compute the left-over variance, defined as
\begin{equation}
\label{eq:leftover_var}
    \lambda_i(t) = 1-\frac{\sum_{j=1}^{i}\sigma_j(t)^2}{\sum_{j=1}^{N}\sigma_j(t)^2}.
\end{equation}
This quantity is shown in figure \ref{fig:svals_Re1250}, and we see that the input-output dynamics of the jet flow at $Re = 1250$ are very low rank since the first few Hankel singular values capture the greatest majority of the variance.

\begin{figure}
\centering
\begin{tikzonimage}[trim= 35 10 70 0,clip,width=0.50\textwidth]{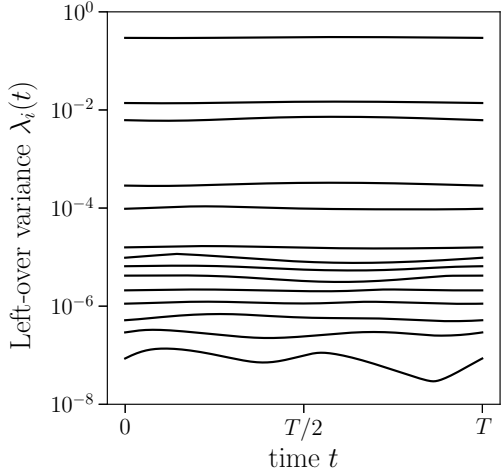}
\end{tikzonimage}
\caption{Left-over variance \eqref{eq:leftover_var}, computed from the Hankel singular values of $\mY(t)^*\mX(t)$ at $Re = 1250$.}
\label{fig:svals_Re1250}
\end{figure}

Moving forward, we select ROM sizes $r = 2$, $r = 4$ and $r = 6$, and we study the performance of these ROMs.
In particular, we wish to see how well they can predict the output $\vy(t)$ in response to (linear) impulses $\mB\vu(t) = \mB\delta(t-\tau)$.
In other words, we compare the ROM to the ground truth obtained from numerical integration of the linearized Navier-Stokes equations.
The outputs from an impulse at time $\tau = 0$ is shown in figure \ref{fig:linear_sensor_output_Re1250}.
Here we see that as we increase the ROM dimension (i.e., as we capture more of the variance), the predictive capabilities of the ROM improve.
In particular, even at $r=4$, the ROM is capable of correctly predicting the amplitude and phase of the response.
With $r=6$, we further improve on the early-time prediction.
For completeness, it is worth mentioning that the ROMs have similar performance also for impulses at times $\tau \neq 0$ (although we do not show the corresponding plots here).

\begin{figure}
\centering
\begin{minipage}{0.48\textwidth}
\begin{tikzonimage}[trim= 5 0 0 0,clip,width=0.85\textwidth]{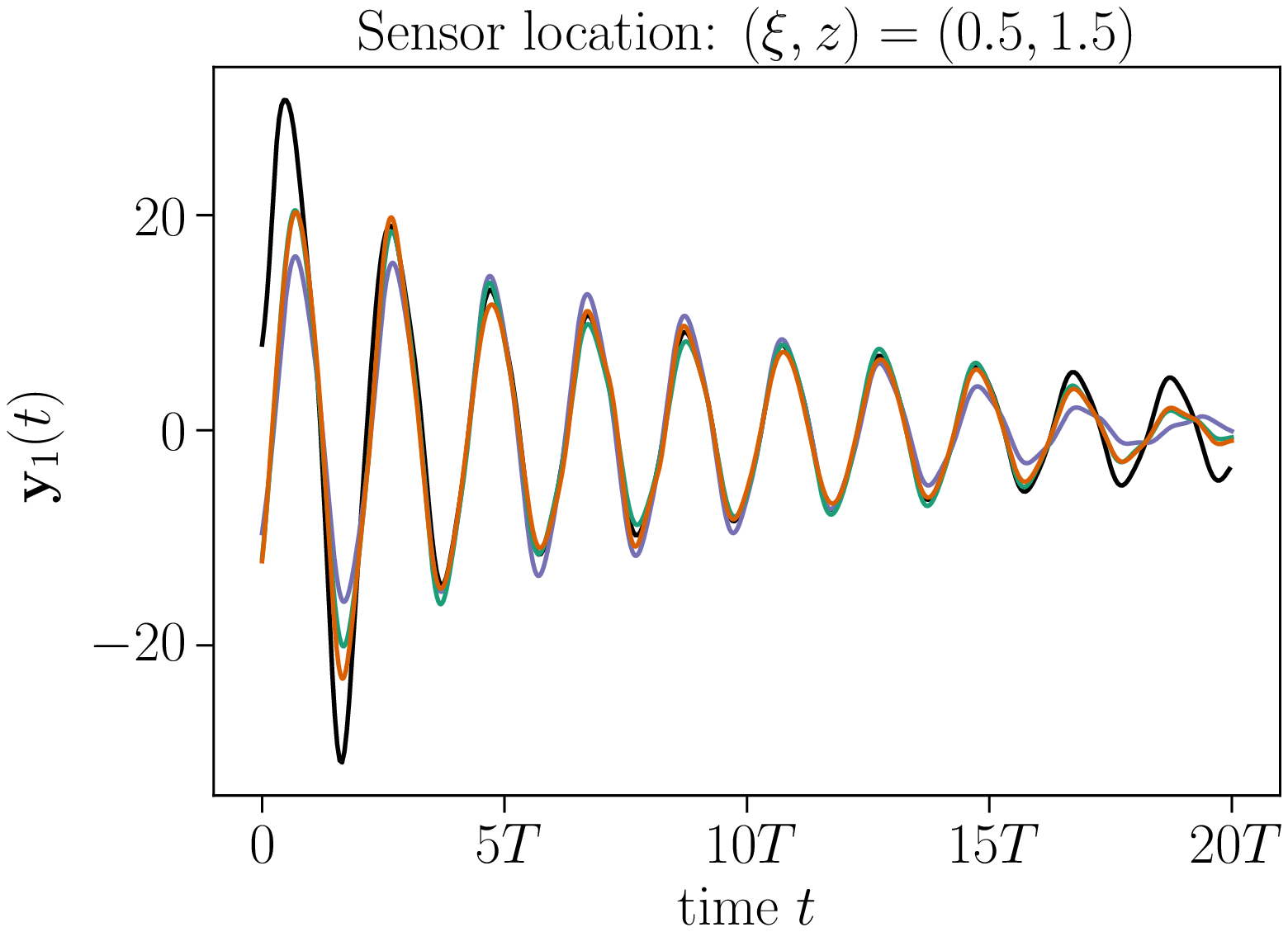}
\node at (0.90,0.8) {\small $\textit{(a)}$};
\end{tikzonimage}
\end{minipage}
\begin{minipage}{0.48\textwidth}
\begin{tikzonimage}[trim= 5 0 0 0,clip,width=0.85\textwidth]{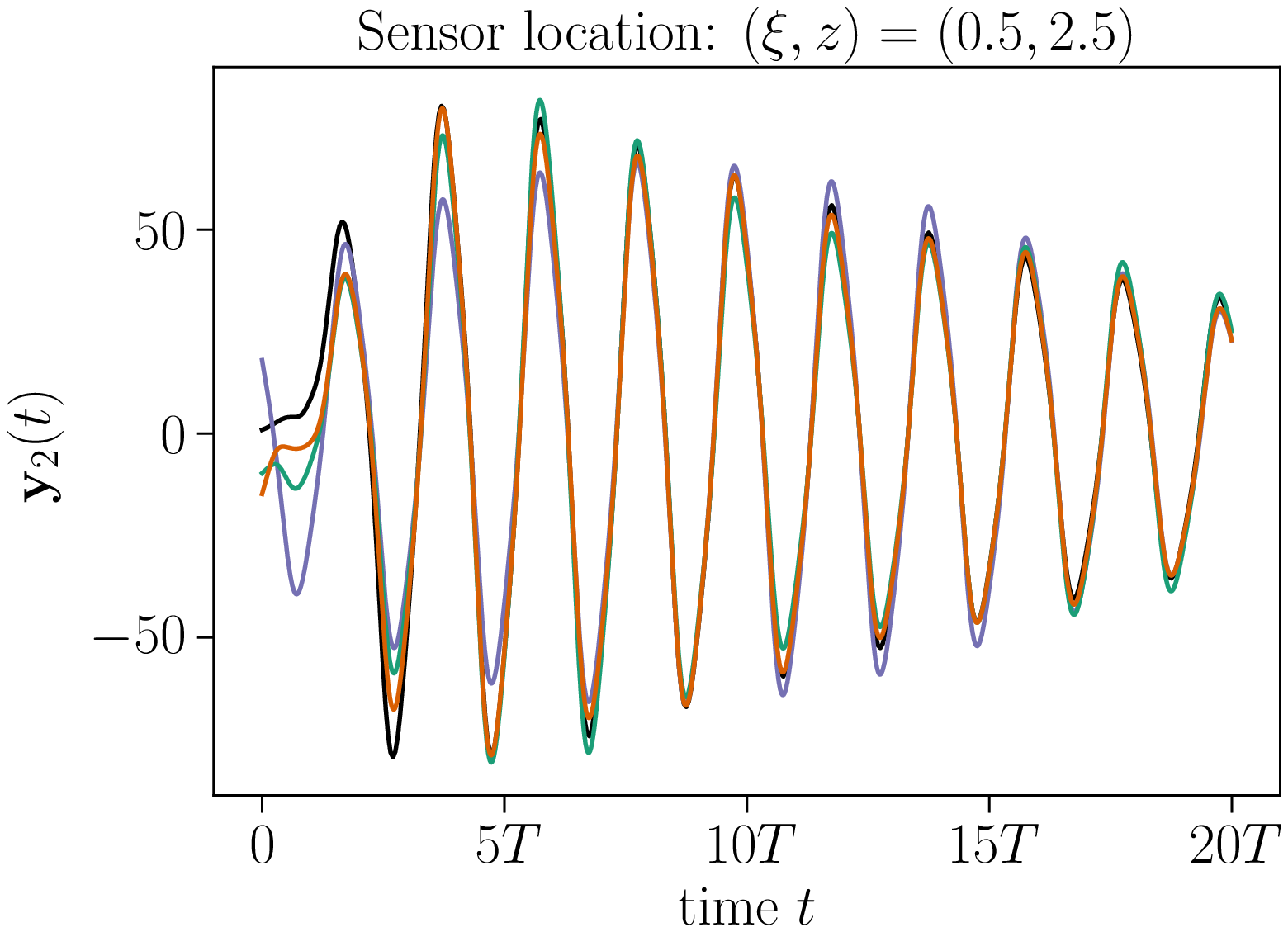}
\node at (0.90,0.8) {\small $\textit{(b)}$};
\end{tikzonimage}
\end{minipage}
\begin{minipage}{0.48\textwidth}
\begin{tikzonimage}[trim= 5 0 0 0,clip,width=0.85\textwidth]{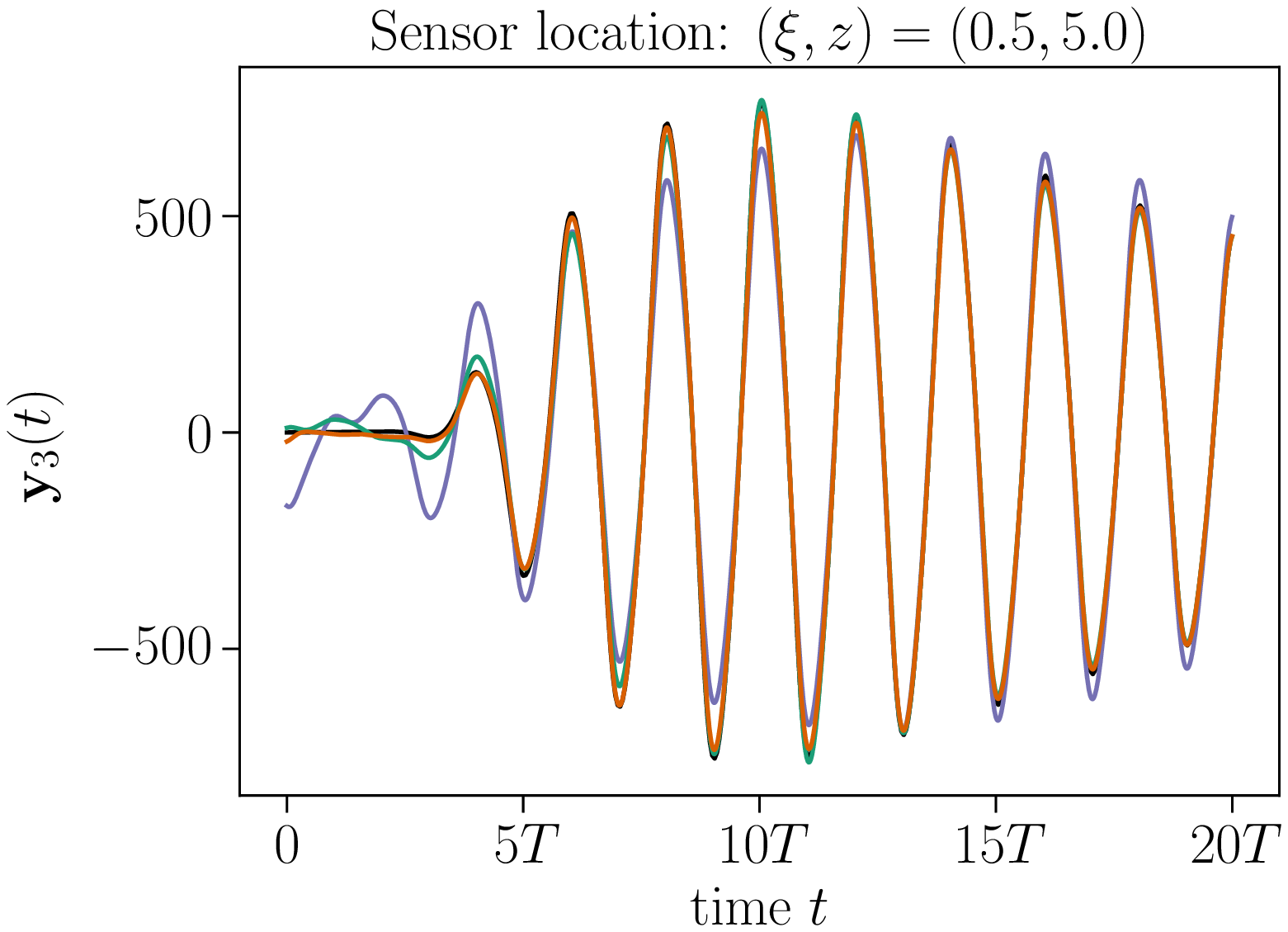}
\node at (0.90,0.8) {\small $\textit{(c)}$};
\end{tikzonimage}
\end{minipage}
\begin{minipage}{0.48\textwidth}
\begin{tikzonimage}[trim= 5 0 0 0,clip,width=0.85\textwidth]{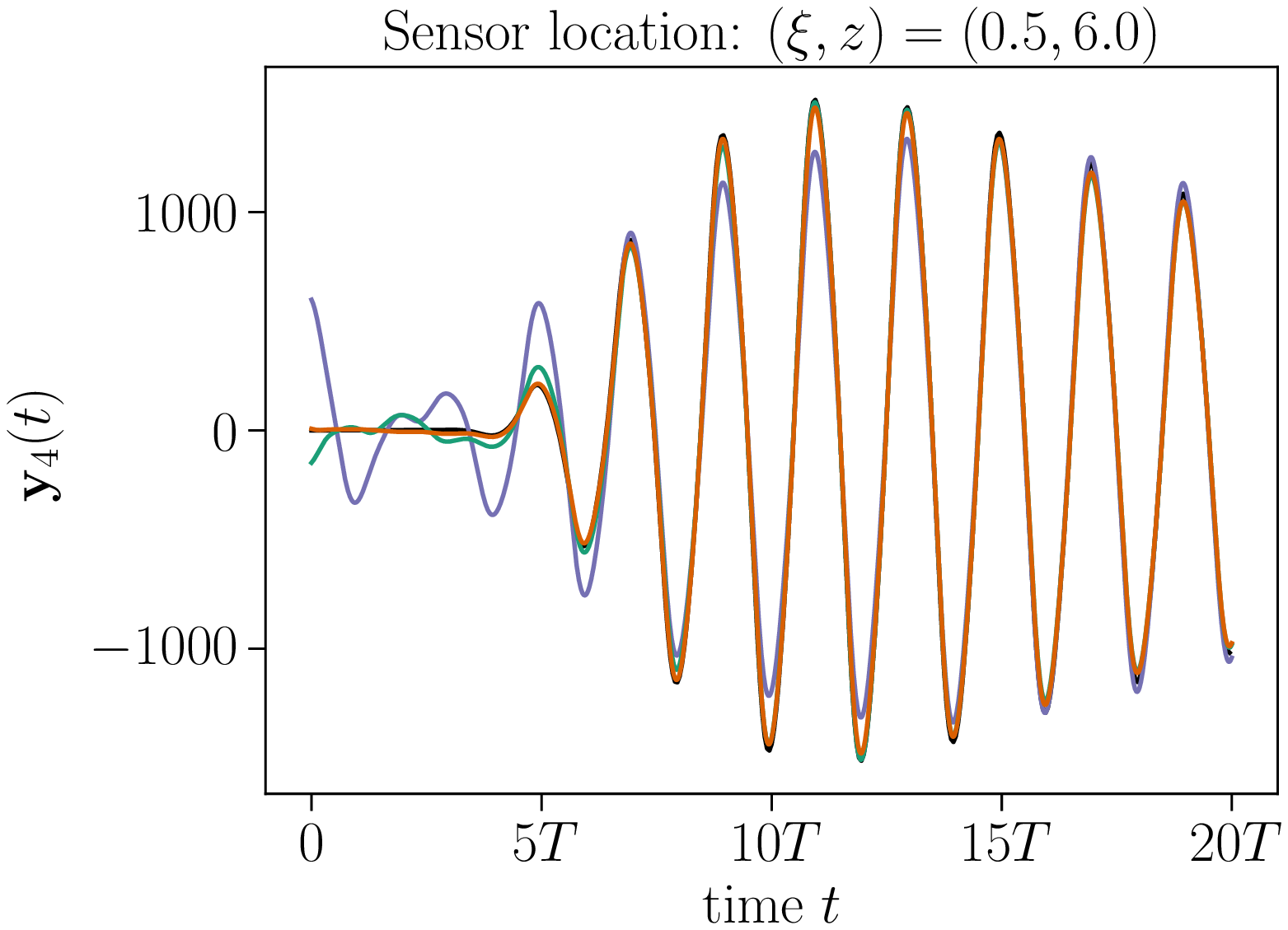}
\node at (0.90,0.8) {\small $\textit{(d)}$};
\end{tikzonimage}
\end{minipage}
\caption{Measured output $\vy_i(t)$ at all four sensor location from the linear impulse response at $Re = 1250$. Ground truth (black), ROM with size $r =2$ (purple), ROM with size $r=4$ (green), ROM with size $r = 6$ (orange).}
\label{fig:linear_sensor_output_Re1250}
\end{figure}

\section{Reduced-order models at $Re = 1500$}
\label{app:roms_Re1500}

Here, we perform the same analysis as in \ref{app:roms_Re1250}, except that we consider $Re = 1500$.
Recall that for this configuration the base flow is linearly unstable.
The left-over variance is shown in figure \ref{fig:svals_Re1500}, while the measured output from a linear impulse response at time $\tau = 0$ is shown in figure \ref{fig:linear_sensor_output_Re1500}. In the latter, we see that the ROMs of sizes $r=4$ and $r=6$ correctly predict the amplitude and phase of the measured outputs, as well as the linear growth rate due to the instability in the underlying base flow.

\begin{figure}
\centering
\begin{tikzonimage}[trim= 35 10 70 0,clip,width=0.50\textwidth]{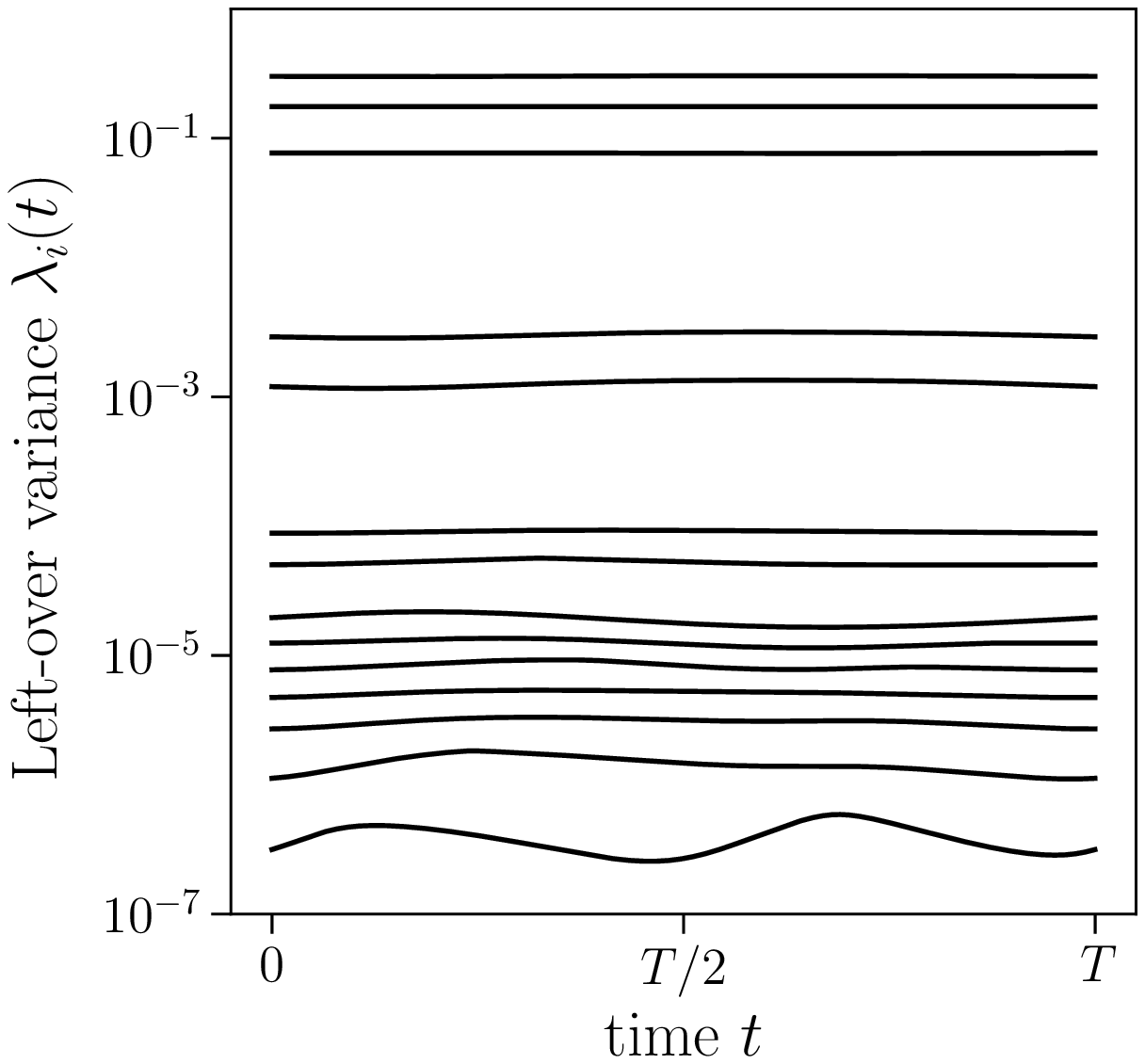}
\end{tikzonimage}
\caption{Analog of figure \ref{fig:svals_Re1250} at $Re = 1500$.}
\label{fig:svals_Re1500}
\end{figure}

\begin{figure}
\centering
\begin{minipage}{0.48\textwidth}
\begin{tikzonimage}[trim= 5 0 0 0,clip,width=0.85\textwidth]{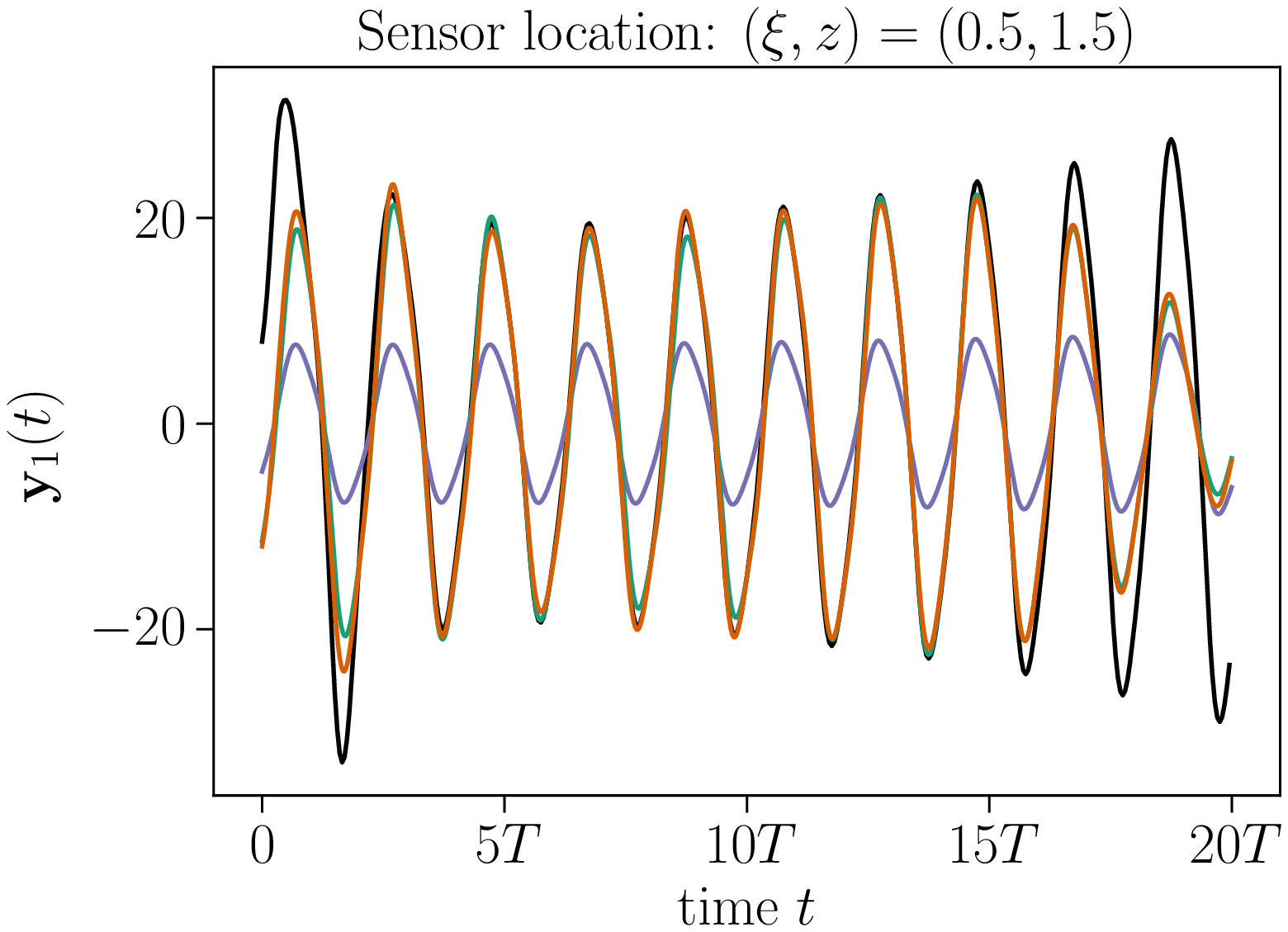}
\node at (0.90,0.8) {\small $\textit{(a)}$};
\end{tikzonimage}
\end{minipage}
\begin{minipage}{0.48\textwidth}
\begin{tikzonimage}[trim= 5 0 0 0,clip,width=0.85\textwidth]{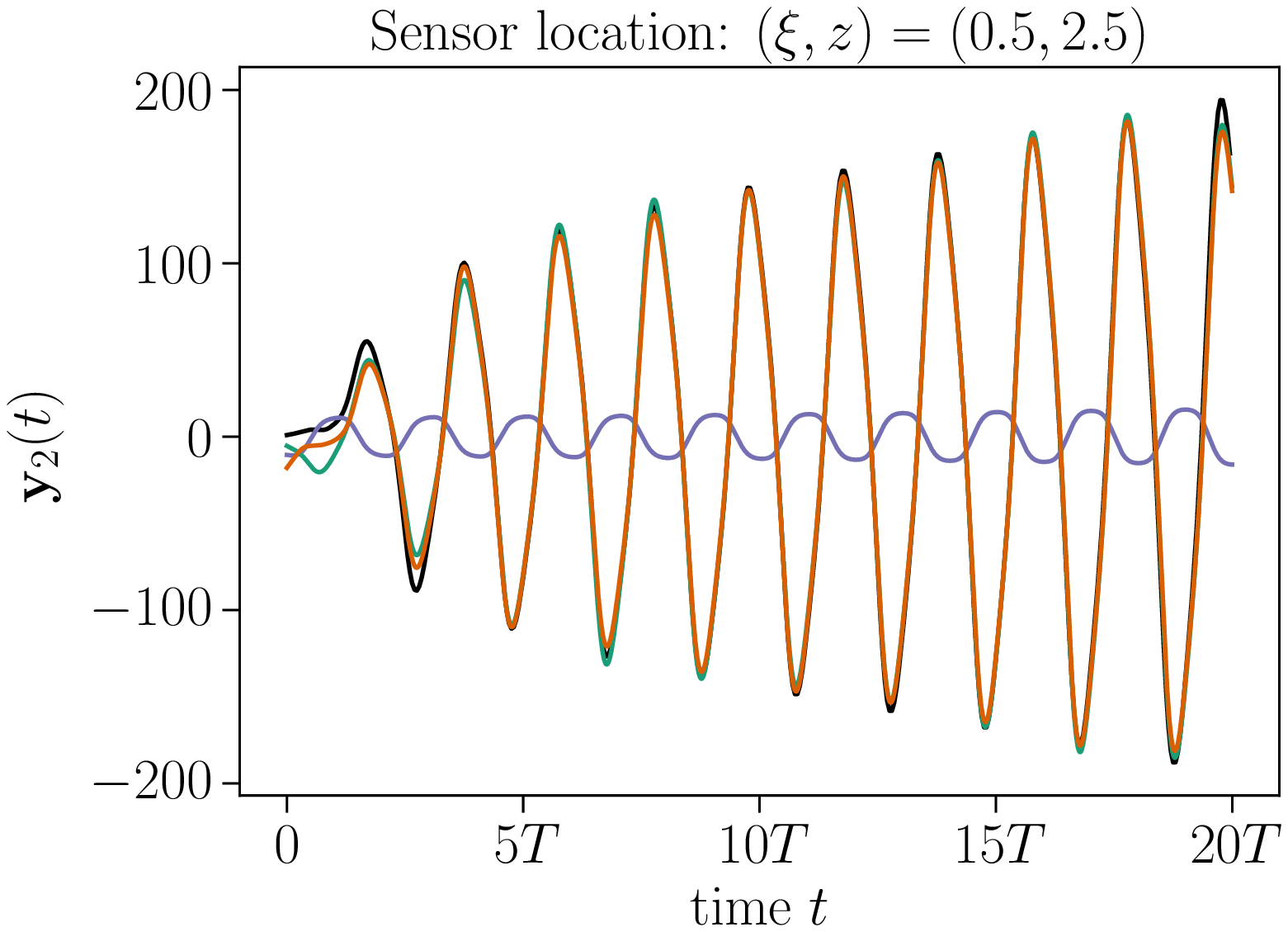}
\node at (0.90,0.8) {\small $\textit{(b)}$};
\end{tikzonimage}
\end{minipage}
\begin{minipage}{0.48\textwidth}
\begin{tikzonimage}[trim= 5 0 0 0,clip,width=0.85\textwidth]{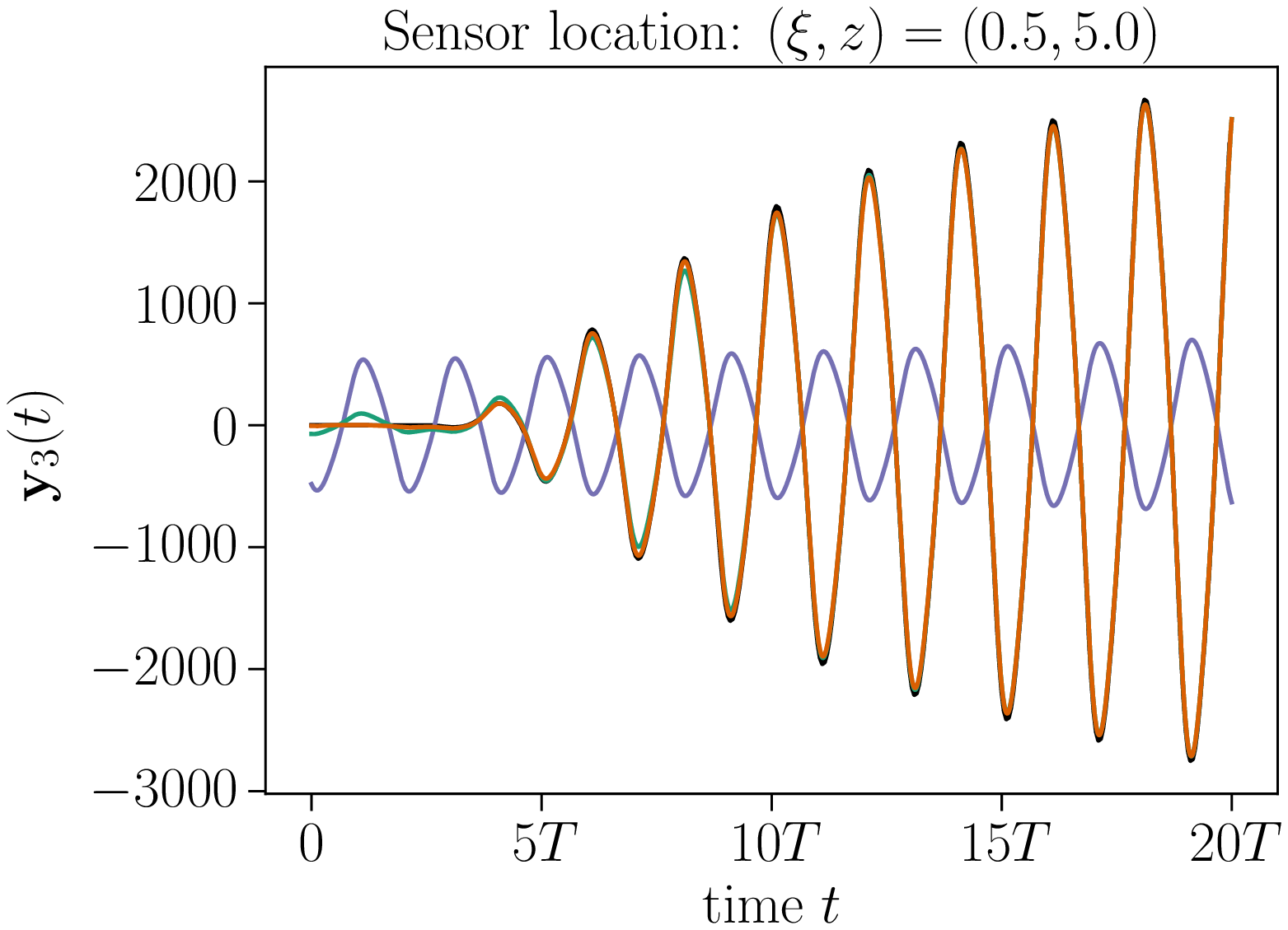}
\node at (0.90,0.8) {\small $\textit{(c)}$};
\end{tikzonimage}
\end{minipage}
\begin{minipage}{0.48\textwidth}
\begin{tikzonimage}[trim= 5 0 0 0,clip,width=0.85\textwidth]{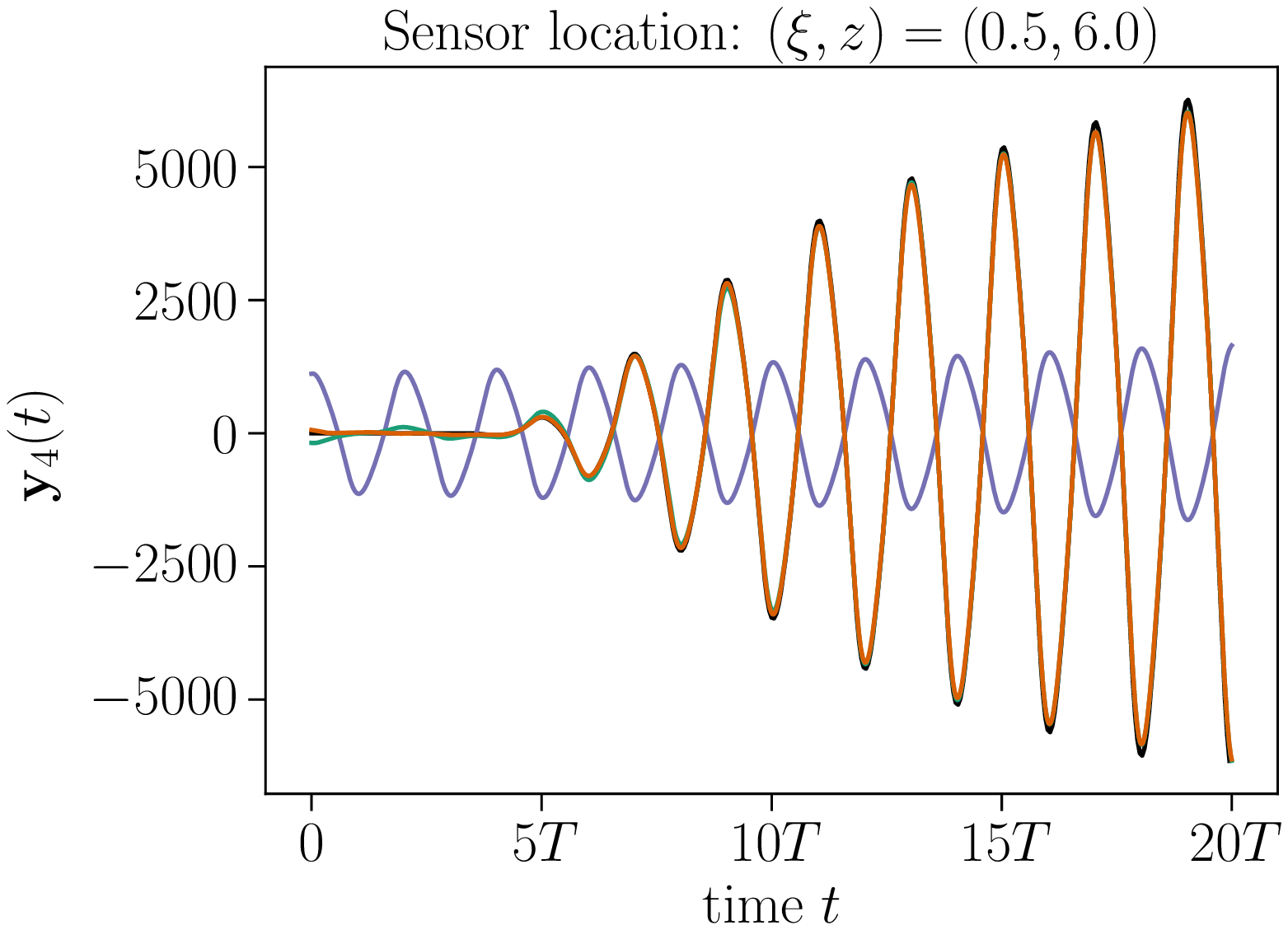}
\node at (0.90,0.8) {\small $\textit{(d)}$};
\end{tikzonimage}
\end{minipage}
\caption{Analog of figure \ref{fig:linear_sensor_output_Re1250} at $Re = 1500$.}
\label{fig:linear_sensor_output_Re1500}
\end{figure}

\bibliographystyle{elsarticle-harv}
\bibliography{references}


\end{document}